%% file: pods2015parallelcorrectness.tex
\documentclass[twocolumn]{article}

\usepackage{amsmath}
\usepackage{amsthm}
\usepackage{amssymb}
\usepackage{hyperref}
\usepackage{verbatim}
\usepackage{xcolor}
\usepackage{xspace}
\usepackage{csquotes}
\usepackage{tabularx}
\usepackage{enumerate}

\input{commands}

\begin{document}
\title{Parallel-Correctness and Transferability for Conjunctive Queries}
\author{ Tom~J.~Ameloot \\ {\small Hasselt University}
  \and Gaetano Geck \\ {\small TU Dortmund University} 
 \and Bas Ketsman \\ {\small Hasselt University} \and Frank Neven \\ {\small Hasselt University}
\and Thomas Schwentick \\ {\small TU Dortmund University}
}

\date{}

\maketitle

\begin{abstract}
A dominant cost for query evaluation in modern massively distributed
systems is the number of communication rounds. For this
reason, there is a growing interest in single-round multiway join
algorithms where data is first reshuffled over many servers and then
evaluated in a parallel but communication-free way. The reshuffling
itself is specified as a distribution policy. We introduce a correctness
condition, called {\em parallel-correctness}, for the evaluation of queries w.r.t. a distribution policy.  We study the complexity of parallel-correctness for conjunctive queries as well as 
transferability of parallel-correctness between queries. We also investigate 
the complexity of transferability for certain families of distribution policies, including, for instance, the {Hypercube} distribution.
\end{abstract}

\input{intro}

\input{definitions}

\input{parallelcorrectness}

\input{transferability}

\input{families}

\input{conclusions}

\subsubsection*{Acknowledgments}
We thank Serge Abiteboul, Luc Segoufin, Cristina Sirangelo, and
Thomas Zeume for helpful remarks.

{\small
\bibliography{references}
\bibliographystyle{abbrv}}

\appendix
\input{appendix}

\end{document}

%% file: commands.tex
\theoremstyle{definition}
\newtheorem{definition}{Definition}[section]
\newtheorem{proposition}[definition]{Proposition}
\newtheorem{question}[definition]{Question}
\newtheorem{example}[definition]{Example}
\newtheorem{remark}[definition]{Remark}
\newtheorem{theorem}[definition]{Theorem}
\newtheorem{lemma}[definition]{Lemma}
\newtheorem{corollary}[definition]{Corollary}
\newtheorem{claim}[definition]{Claim}

\newcommand{\problemdefinition}[3]
	{%
		\medskip%
		\noindent%
		\begin{tabularx}{\textwidth}{lX}%
			\textbf{Problem:}  & #1 \\%
			\textbf{Input:}    & #2 \\%
			\textbf{Question:} & #3 \\%
		\end{tabularx}%
		\medskip%
	}

\newcommand{\probdefbodynew}[3]{
\medskip
\noindent
\fbox{\parbox{\columnwidth}{
#1:\\
{\bf Input:} #2\\
{\bf Question:} #3 }}
\smallskip
}

\newcommand{\probdefbodybb}[4]{
\medskip
\noindent
\fbox{\parbox{\columnwidth}{
#1:\\
{\bf Input:} #2\\
{\bf Black box input:} #3\\
{\bf Question:} #4 }}
\smallskip
}

\newenvironment{case}[2]
	{%
		\medskip%
		\noindent%
		\textbf{Case #1 (#2):} %
	}
	{%
		\smallskip%
	}

\newcommand{\EM}[1]{\ensuremath{#1}\xspace} 
\newcommand{\newabbrev}[2]{\newcommand{#1}{{#2}\xspace}}
\newcommand{\newmathabb}[2]{\newcommand{#1}{\EM{{#2}}}}
\newcommand{\newtuple}[2]{\newmathabb{#1}{\mathbf{#2}}}

\newcommand{\img}[1]{\EM{\mathit{img}(#1)}} 
\newcommand{\set}[1]{\EM{\{#1\}}} 

\newcommand{\dom}{\EM{\mathbf{dom}}} 
\newcommand{\domk}{\EM{\mathit{dom}_k}} 
\newcommand{\adom}[1]{\EM{\mathit{adom}(#1)}} 
\newcommand{\facts}[1]{\EM{\mathit{facts}(#1)}} 
\newcommand{\fcX}[1]{\EM{\boldsymbol{#1}}} 
\newcommand{\fc}{\EM{\fcX{f}}} 
\newcommand{\fhead}{\fcX{h}} 
\newcommand{\fcB}{\EM{\fcX{g}}} 
\newcommand{\qr}{\EM{\mathcal{Q}}} 
\newcommand{\sch}{\EM{\mathcal{D}}} 
\newcommand{\uvar}{\EM{\mathbf{var}}} 
\newcommand{\ar}[1]{\EM{\mathit{ar}(#1)}} 
\newcommand{\tup}[2][]{\EM{\mathbf{#2}_{#1}}} 

\newcommand{\distp}{\EM{\boldsymbol{P}}} 
\newcommand{\distvar}[2]{\EM{\mathit{dist}_{#1}(#2)}}
\newcommand{\dist}[1]{\distvar{\distp}{#1}}
\newcommand{\nw}{\EM{\mathcal{N}}} 

\newcommand{\atomB}{\EM{B}} 
\newcommand{\body}[1]{\EM{\mathit{body}_{#1}}} 
\newcommand{\head}[1]{\EM{\mathit{head}_{#1}}} 
\newcommand{\vars}[1]{\EM{\mathit{vars}(#1)}} 
\newcommand{\simp}{\EM{\theta}} 
\newcommand{\mathpredicate}[1]{\EM{\mathtt{#1}}}
\newcommand{\newpredicate}[2]{\newcommand{#1}{\EM{\mathpredicate{#2}}}}
\newcommand{\mathproblem}[1]{\EM{\textsc{#1}}\xspace}
\newcommand{\newproblem}[2]{\newcommand{#1}{\EM{\mathproblem{#2}}}}

\newcommand{\polyh}[1]{\EM{\Pi_{#1}^{P}}}
\newcommand{\copolyh}[1]{\EM{\Sigma_{#1}^{P}}}
\newcommand{\phtwo}{\polyh{2}}
\newcommand{\phthree}{\polyh{3}}
\newcommand{\cophtwo}{\copolyh{2}}
\newcommand{\cophthree}{\copolyh{3}}
\newcommand{\pik}[1]{\EM{\Pi_{#1}}}
\newcommand{\pitwo}{\pik{2}}
\newcommand{\pithree}{\pik{3}}

\newcommand{\coNP}{\EM{\textsc{coNP}}}
\newcommand{\NP}{\EM{\textsc{NP}}}
\newcommand{\coClass}[1]{\textsc{co-}{#1}}

\newproblem{\Sat}{Sat}
\newproblem{\ThreeSat}{3-Sat}
\newproblem{\PC}{PC}
\newproblem{\PCI}{PCI}
\newproblem{\SMQ}{StronglyMinimal}
\newproblem{\PCTRANS}{{\sc pc-trans}}

\newmathabb{\alg}{\mathcal{A}} 

\providecommand{\simp}{\theta} 

\newcommand{\classpnrel}{\ensuremath{{\cal P}_{\text{nrel}}}}

\newcommand{\classpfin}{{\cal P}_{\text{fin}}}

\renewcommand{\phi}{\varphi}
\newabbrev{\formulae}{formul\ae}
\newabbrev{\subformulae}{subformul\ae}

\newcommand{\inducedminimal}{minimal\xspace}

\newabbrev{\im}{\inducedminimal}
\newabbrev{\imity}{minimality}
\newabbrev{\Imity}{Minimality}
\newabbrev{\imv}{\im valuation}
\newabbrev{\animv}{a \im valuation}
\newabbrev{\Ta}{Truth assignment}
\newabbrev{\ta}{truth assignment}
\newabbrev{\tas}{truth assignments}

\newcommand{\structure}[2]{\EM{\mathit{Struct}(#2)}}
\newcommand{\consistency}[1]{\EM{\mathit{Cons}}}
\newmathabb{\domain}{\mathit{Values}}

\newmathabb{\vx}{\mathbf{x}}
\newmathabb{\vy}{\mathbf{y}}
\newmathabb{\vz}{\mathbf{z}}
\newmathabb{\nxi}{\overline{\xi}}
\newmathabb{\nx}{\overline{x}}
\newmathabb{\ny}{\overline{y}}
\newmathabb{\nz}{\overline{z}}
\newmathabb{\negu}{\overline{u}} 
\newmathabb{\negx}{\neg{x}}
\newmathabb{\negy}{\neg{y}}
\newmathabb{\negz}{\neg{z}}
\newcommand{\lit}[2]{\EM{\ell_{#1,#2}}}

\newmathabb{\seqc}{C_1,\dots,C_\lc}
\newmathabb{\seqx}{x_1,\dots,x_\lx}
\newmathabb{\seqy}{y_1,\dots,y_\ly}
\newmathabb{\seqz}{z_1,\dots,z_\lz}
\newmathabb{\seqnx}{\nx_1,\dots,\nx_\lx}
\newmathabb{\seqny}{\ny_1,\dots,\ny_\ly}
\newmathabb{\seqnz}{\nz_1,\dots,\nz_\lz}
\newcommand{\seqlit}[1]{\EM{\lit{j}{1},\lit{j}{2}, \lit{j}{3}}}

\newmathabb{\itx}{g}
\newmathabb{\itax}{g'}
\newmathabb{\ity}{h}
\newmathabb{\itz}{i}
\newmathabb{\itc}{j}
\newmathabb{\Itx}{G}
\newmathabb{\Itc}{J}

\newmathabb{\lx}{m}
\newmathabb{\ly}{n}
\newmathabb{\lz}{p}
\newmathabb{\lc}{k}
\newmathabb{\li}{q}
\newmathabb{\indx}{1,\dots,\lx}
\newmathabb{\indy}{1,\dots,\ly}
\newmathabb{\indz}{1,\dots,\lz}
\newmathabb{\indc}{1,\dots,\lc}

\newcommand{\disClause}[1]{\EM{(\lit{#1}{1} \lor \lit{#1}{2} \lor \lit{#1}{3})}}

\newmathabb{\varf}{w_0}
\newmathabb{\vart}{w_1}
\newtuple{\vb}{b}
\newtuple{\vu}{u}
\newtuple{\vw}{w}

\newpredicate{\Neg}{Neg}
\newpredicate{\True}{True}
\newpredicate{\False}{False}
\newpredicate{\Dom}{Val}

\newcommand{\TQBF}[1]{\EM{{#1}\mathproblem{-QBF}}}

\newcommand{\twoTQBF}{\TQBF{\pitwo}}
\newcommand{\threeTQBF}{\TQBF{\pithree}}
\newproblem{\pitwoTQBF}{\TQBF{\pitwo}}
\newproblem{\pithreeTQBF}{\TQBF{\pithree}}
\newproblem{\pivarTQBF}{\TQBF{\pik{i}}}

\newmathabb{\uniV}{V}
\newmathabb{\exisV}{V'}
\newmathabb{\uniS}{s_V}
\newmathabb{\exisS}{s_{V'}}

\newmathabb{\minVal}{V^\ast}
\newmathabb{\val}{V}
\newmathabb{\altVal}{V'}

\newcommand{\as}{{\mathcal{A}}} 

\newcommand{\sbucket}[1]{\mathit{bucket}_{#1}} 
\newcommand{\abucket}[1]{\mathit{bucket}^*_{#1}} 

\newpredicate{\Clause}{C}
\newmathabb{\dquery}{\qr_\phi}
\newmathabb{\dinstance}{I_\phi}
\newmathabb{\dpolicy}{\distp_\phi}

\newmathabb{\Bool}{\mathbb{B}}
\newmathabb{\negBool}{\Bool^-}
\newmathabb{\posBool}{\Bool^+}

\newmathabb{\Var}{\mathbb{W}}
\newmathabb{\negVar}{\Var^-}
\newmathabb{\posVar}{\Var^+}

\newmathabb{\VarAlt}{\mathbb{U}}
\newmathabb{\negVarAlt}{\VarAlt^-}
\newmathabb{\posVarAlt}{\VarAlt^+}

\newmathabb{\threeBool}{\Bool}
\newmathabb{\threePosBool}{\posBool}
\newmathabb{\threeNegBool}{\negBool}

\newmathabb{\threeVar}{\Var}
\newmathabb{\threePosVar}{\posVar}
\newmathabb{\threeNegVar}{\negVar}

\newmathabb{\sixVar}{\VarAlt}
\newmathabb{\sixPosVar}{\posVarAlt}
\newmathabb{\sixNegVar}{\negVarAlt}

\newmathabb{\Lit}{\mathbb{L}}
\newmathabb{\sixLit}{\Lit}

\newmathabb{\badVal}{V}
\newmathabb{\goodVal}{V'}

\newcommand{\Fix}{\text{Fix}}
\newcommand{\fix}{\Fix}

\newmathabb{\negFacts}{I^-_\phi}
\newmathabb{\necFacts}{I^+_\phi}
\newmathabb{\misFact}{\fc}
\newmathabb{\ssI}{s_{I'}}

\newmathabb{\rep}{\mathit{rep}}

\newcommand{\nodeA}{\EM{\kappa^+}}
\newcommand{\nodeB}{\EM{\kappa^-}}
\newcommand{\node}{\EM{\kappa}}
\newmathabb{\nodeset}{N}

\newcommand{\mydef}{\ensuremath{\mathrel{\smash{\stackrel{\scriptscriptstyle{
    \text{def}}}{=}}}}}
\newcommand{\mynewr}[1][0mm]{\\[#1]\mbox{}\hfill}

%% file: intro.tex
\section{Introduction}
\label{sec:intro}

In traditional database systems, the complexity of
query processing for large datasets is mainly determined by the number of 
IO requests to external memory. A factor dominating complexity in modern massively distributed database systems, however, is the number of communication steps~\cite{DBLP:conf/pods/BeameKS14}.
Motivated by recent in-memory systems like Spark~\cite{spark} and Shark~\cite{shark}, Koutris and Suciu introduced the massively parallel communication model (MPC)~\cite{DBLP:conf/pods/KoutrisS11} where computation proceeds in a sequence of parallel steps each followed by global synchronization of all servers. In this model, evaluation of conjunctive queries~\cite{DBLP:conf/pods/BeameKS13,DBLP:conf/pods/KoutrisS11} and skyline queries~\cite{DBLP:conf/icdt/AfratiKSU12} has been considered. 

Of particular interest in the MPC model are the queries that can be evaluated in one round of communication. Recently, Beame, Koutris and Suciu~\cite{DBLP:conf/pods/BeameKS14} proved a matching upper and lower bound for the amount of communication needed to compute a full conjunctive query without self-joins in one communication round. The upper bound is provided by a randomized algorithm called {Hypercube} which uses a technique that can be traced back to Ganguli, Silberschatz, and Tsur~\cite{DBLP:journals/jlp/GangulyST92} and is described in the context of map-reduce by Afrati and Ullman~\cite{AfratiUllman10}. 
The Hypercube algorithm evaluates a conjunctive query $\qr$ by first reshuffling the data over many servers and then evaluating $\qr$ at each server in a parallel but communication-free manner. The reshuffling is
specified by a distribution policy (hereafter, called Hypercube distribution) and is based on the structure of $\qr$. 
In particular, the Hypercube distribution partitions the space of all complete valuations of $\qr$ over the computing servers in an instance independent way through hashing of domain values. 
A property of Hypercube distributions is that for any instance $I$, the central execution of $\qr(I)$ always equals the union of the evaluations of $\qr$ at every computing node (or server).\footnote{We note that, for a query $\qr$, there is no single Hypercube distribution but rather a family of distributions as the concrete instantiation depends on choices regarding the address space of servers.}

In this paper, we introduce a general framework for reasoning about one-round evaluation algorithms under {\em arbitrary} distribution policies. 
Distribution policies (formally defined in Section~\ref{sec:defs}) are functions mapping input facts to sets of nodes (servers) in the network. We introduce the following correctness property for queries and distribution policies: a query $\qr$ is {\em parallel-correct} for a given distribution policy $\distp$, when for any instance $I$, the evaluation of $\qr(I)$ equals the union of the evaluation of $\qr$ over the distribution of $I$ under policy $\distp$. 
We focus on conjunctive queries and study the complexity of 
deciding parallel-correctness. We show that the latter problem is equivalent to
testing whether the facts in every \emph{minimal} valuation of the conjunctive query are mapped to a same node in the network by the distribution policy. For various representations of distributions policies, we then show that testing parallel-correctness is in \polyh{2}. We provide a 
matching lower bound via a reduction from the $\phtwo$-complete
  $\twoTQBF$-problem.

One-round evaluation algorithms, like Hypercube, redistribute data for the evaluation of every query.  For scenarios where queries are executed in sequence, it makes sense to study cases where the same data distribution can be used to evaluate multiple queries. We formalize this as 
{\em parallel-correctness transfer} between queries. In particular, parallel-correctness {\em transfers} from $\qr$ to $\qr'$ when $\qr'$ is parallel-correct under every distribution policy for which $\qr$ is parallel-correct. We characterize transferability for conjunctive queries by a (value-based) containment condition for minimal valuations of $\qr'$ and $\qr$,
and use this characterization to obtain a $\phthree$ upper bound for transferability. Again, we obtain a matching lower bound, this time via a reduction from the $\phthree$-complete $\threeTQBF$-problem. We obtain a (presumably) better complexity, \NP-completeness, in the case that $\qr$ is \emph{strongly minimal}, i.e., when all its valuations are minimal.  Examples of strongly minimal CQs include the full conjunctive queries and those without self-joins. At the heart of the upper bound proof lies the insight that the above mentioned value-based inclusion w.r.t.\ minimal valuations reduces to a syntactic inclusion of $\qr'$ in $\qr$ modulo a variable renaming when $\qr$ is strongly minimal. We obtain that deciding strong minimality is \NP-complete as well. 

Finally, we study parallel-correctness transfer from $\qr$ to $\qr'$ w.r.t.\ a specific family of distribution policies $\cal F$ rather than the set of {\em all} distribution policies. We show that 
it is \NP-complete to decide whether $\qr'$ is parallel-correct for a given
family $\cal F$ if this family has the following two properties: it is $\qr$-generous (for each, not only for minimal, valuation of $\qr$, its facts occur at some node) and  $\qr$-scattered (for every instance some distribution has, at every node, only facts from one valuation). It is easy to see that the family  of Hypercube distributions for a given CQ $\qr$ satisfies these properties, which implies that deciding transferability for Hypercube distributions is \NP-complete, as well. 

We complete our framework by sketching a declarative specification formalism for distribution policies, illustrated with the specification of Hypercube distributions.

Due to space restrictions, many proofs are moved to an appendix.

\bigskip
\noindent
{\bf Outline.} We introduce the necessary definitions in Section~\ref{sec:defs}. We study parallel-correctness in Section~\ref{sec:parallel-correctness} and transferability in Section~\ref{sec:trans}. We examine families of distribution policies including the {Hypercube} distribution in Section~\ref{sec:families}. 
We conclude in Section~\ref{sec:conclusion}.


%% file: definitions.tex

\section{Definitions}
\label{sec:defs}

\paragraph{Queries and instances.}
We assume an infinite set \dom of data values that can be represented
by strings over some fixed alphabet. 
A \emph{database schema} $\sch$ is a finite set of relation names $R$ where every $R$ has arity $\ar R$. We call $R(\tup t)$ a \emph{fact} when $R$ is a relation name and $\tup t$ a tuple in \dom.
We say that a fact $R(d_1, \ldots, d_k)$ is {\em over} a database
schema $\sch$ if $R \in \sch$ and $\ar R = k$. By $\facts\sch$, we
denote the set of possible facts over schema $\sch$. 
A \emph{(database) instance} $I$ over $\sch$ is a finite set of facts
over $\sch$.  By $\adom{I}$ we denote the set of data values occurring
in $I$. 
A \emph{query $\qr$ over input schema $\sch_1$ and output schema $\sch_2$} is a generic mapping from instances over $\sch_1$ to instances over $\sch_2$.  Genericity means that for every permutation $\pi$ of \dom\ and every instance $I$, $\qr(\pi(I)) = \pi(\qr(I))$.

\paragraph{Conjunctive queries.}
Let $\uvar$ be the universe of variables, disjoint from $\dom$.
An \emph{atom} is of the form $R(\tup{x})$, where $R$ is a relation name and $\tup{x}$ is a tuple of variables in $\uvar$. We say that $R(x_1,\ldots,x_k)$ is an atom \emph{over} schema $\sch$ if $R \in \sch$ and $k = \ar R$. 

A \emph{conjunctive query} $\qr$ (CQ) over input schema $\sch$ is an expression of the form
\[
T(\tup{x}) \gets R_1(\tup[1] y), \ldots, R_n(\tup[n] y),
\]
where every $R_i(\tup[i] y)$ is an atom over $\sch$, and $T(\tup x)$ is an atom for which $T \not \in \sch$.
Additionally, for safety, we require that every variable in $\tup x$
occurs in some $\tup[i] y$. We refer to the \emph{head atom} $T(\tup
x)$ by $\head{\qr}$, and denote the set of body atoms $R_i(\tup[i]
y)$ by $\body{\qr}$.

A conjunctive query is called \emph{full} if all variables of the body also occur in the head. We say that a CQ is \emph{without self-joins} when all of its atoms have a distinct relation name.

We denote by $\vars{\qr}$ the set of all variables occurring in $\qr$.
A \emph{valuation} for a conjunctive query is a total function $V:\vars{\qr} \to \dom$ that maps each variable of $\qr$ to a data value.
We say that $V$ \emph{requires} or \emph{needs} the facts $V(\body{\qr})$ for $\qr$.
A valuation $V$ is said to be \emph{satisfying} for $\qr$ on instance
$I$, when all the facts required by $V$ for $\qr$ are in $I$.
In that case, $V$ \emph{derives} the fact $V(\head{\qr})$. The result
of $\qr$ on instance $I$, denoted $\qr(I)$, is defined as the set of
facts that can be derived by satisfying valuations for $\qr$ on
$I$. We note that, as we do not allow negation, all conjunctive queries are monotone.

We frequently compare different valuations for a query \qr with
respect to their required sets of facts. For two valuations $V_1,V_2$
for a CQ \qr, we write $V_1\le_Q V_2$ if
$V_1(\head{\qr})=V_2(\head{\qr})$ and $V_1(\body{\qr})\subseteq
V_2(\body{\qr})$. We write  $V_1<_Q V_2$ if furthermore $V_1(\body{\qr})\subsetneq
V_2(\body{\qr})$ holds.

A \emph{substitution} is a mapping from variables to variables,
which is generalized to tuples, atoms and
conjunctive queries in the natural fashion
\cite{ahv_book}.\footnote{As we only consider CQs without constants,
  substitutions do not map variables to constants.} We denote the composition of functions in
the usual way, i.e., $(f\circ g)(x)\mydef f(g(x))$.

The following notion is fundamental for the development in the rest of the paper:
\begin{definition}
  A \emph{simplification} of a conjunctive query $\qr$ is a 
  substitution $\simp:\vars\qr\to\vars\qr$ for which
  $\head{\simp(\qr)}=\head\qr$ and $\body{\simp(\qr)}\subseteq\body\qr$.
\end{definition}
    A simplification is thus a homomorphism from $\qr$ to $\qr$ and by the homomorphism theorem \cite{ahv_book} (and the trivial embedding from  $\theta(\qr)$ to \qr), \qr and $\simp(\qr)$ are equivalent. Of course, the identity substitution is always a simplification.

\begin{example}\label{ex:subs}
    We give a few examples to illustrate simplifications.
    Consider the query
    \[
        T(x) \gets R(x,x), R(x,y), R(x,z).
    \]
    Then $\simp_1=\set{x\mapsto x, y\mapsto y, z\mapsto y}$ as well as   
    $\simp_2=\set{x\mapsto x, y\mapsto x, z\mapsto x}$ are simplifications.
    For the query
    \[
        T(x) \gets R(x,y), R(y,y), R(z,z), R(u,u),
    \]
    possible simplifications are $\simp_3 = \set{x\mapsto x, y\mapsto y, z\mapsto y, u\mapsto z}$ and $\simp_4 = \set{x\mapsto x, y\mapsto y, z\mapsto y, u\mapsto y}$.
    For the query $T(x) \gets R(x,y), R(y,z)$
        there are no simplifications besides the identity.
    \qed
\end{example}

The notion of simplification is closely related to foldings as defined by Chandra and Merlin~\cite{DBLP:conf/stoc/ChandraM77}. In particular, a {\em folding} of a conjunctive query $\qr$ is a simplification $\theta$ that is idempotent. That is, $\theta^2=\theta$. Intuitively, the idempotence means that when $\theta$ gives a new name to a variable then it sticks to it. Notice that in Example~\ref{ex:subs} simplifications $\theta_1$, $\theta_2$, $\theta_4$ are foldings but $\theta_3$ is not as $\theta_3(u)=z\neq y= \theta_3(\theta_3(u))$.

\paragraph{Networks, data distribution, and policies.}
A \emph{network} $\nw$ is a nonempty finite set of values from \dom,
which we call \emph{(computing) nodes}. 

A \emph{distribution policy} $\distp$ for a database schema $\sch$ and
a network $\nw$ is a total function mapping facts from  $\facts\sch$
to sets of nodes.\footnote{Notice that our formalization allows to `skip' facts by mapping them 
to the empty set of nodes. This is, for instance, the case for a {Hypercube} distribution
(cf.\ Section~\ref{sec:families}),  which skips facts that are not essential to evaluate
the query at hand.}
For an instance $I$ over $\sch$, let $\dist I$ denote the function that maps each
$\kappa\in\nw$ to $\{\fc\in I\mid \kappa\in\distp(\fc)\}$, that is, the set of facts assigned to it
by $\distp$. We sometimes refer to 
$\dist{I}(\kappa)$ as a {\em data chunk}.

In this paper, we do not always explicitly give names to schemas and networks but  
tacitly assume they are understood from the queries and the
distribution policies under consideration, respectively.

We do \emph{not} always expect that distribution policies \distp are given 
as part of the input by exhaustive enumeration of all pairs $(\node,\fc)$,
for which $\node\in\distp(\fc)$. We also consider
mechanisms, where instead the distribution policy is
implicitly represented by a given ``black box'' procedure. 
While there are many possible ways to represent distribution policies, either as functions or as relations belonging to various complexity classes, in this paper, we only consider one such class. In particular, we define the class \classpnrel\ where each distribution \distp is represented by a $\NP$-testable relation, that on input $(\node,\fc)$ yields ``true'' if and only if $\node\in\distp(\fc)$.
We will discuss declarative ways to specify distribution policies in a
non-black-box fashion in Section \ref{sec:families}.

The definition of a distribution policy is borrowed from Ameloot et al.~\cite{DBLP:conf/pods/AmelootKNZ14} (but already surfaces in the work of Zinn et al.~\cite{DBLP:conf/icdt/ZinnGL12}), where distribution policies are used to define the class of policy-aware transducer networks.


%% file: parallelcorrectness.tex
\section{Parallel-Correctness}
\label{sec:parallel-correctness}

In this section, we introduce and study the notion of parallel-correctness, which
is central to this paper.

\begin{definition}
	\label{def:pc-inst}
A query \qr is \emph{parallel-correct on instance} $I$ {\em under  
distribution policy} \distp, if $\qr(I) = \bigcup_{\node\in\nw}{\qr(\dist I(\node))}$.
\end{definition}

That is, the centralized execution of $\qr$ on $I$ is the same as taking the union of the results obtained by executing $\qr$ at every computing node.
Next, we lift parallel-correctness to all instances.

\begin{definition}
\label{def:pc:general}
A query $\qr$ is \emph{parallel-correct under distribution policy} $\distp$, if $\qr$ 
is parallel-correct on all input instances under \distp. 
\end{definition}

Of course, when a query $\qr$ is parallel-correct under $\distp$, there is a direct one-round evaluation algorithm for every instance. Indeed, the algorithm first distributes (reshuffles) the data over the computing nodes according to $\distp$ and then evaluates $Q$ in a subsequent parallel step at every computing node. Notice that as $\distp$ is defined on the granularity of a fact, the reshuffling does not depend on the current distribution of the data and can be done in parallel as well.

{Although Definitions~\ref{def:pc-inst} and \ref{def:pc:general}
are in terms of general queries, in the rest of this section, we only consider conjunctive queries.} It is easy to see that a CQ $\qr$ is parallel-correct under distribution policy $\distp$ if for each valuation for $\qr$ the
required facts meet at some node, i.e., if the following condition holds:
\begin{itemize}
\item[(C0)] for every valuation $V$ for \qr,
  $$\bigcap_{\fc \in V(\body{\qr})}\distp(\fc) \ne \emptyset.$$
\end{itemize}

Even though (C0) is sufficient for parallel-correctness,
it is not necessary {(c.f., Example~\ref{ex:minimal_valuations})}.
It turns out that for a
semantical characterization only valuations have to be considered
that are minimal in the following sense. 
\begin{definition}
    Let $\qr$ be a CQ. A valuation $V$ for $\qr$ is \emph{\inducedminimal} for $\qr$ if
    there exists no valuation $V'$ for $\qr$ such that $V'<_\qr V$.
\end{definition}

The next lemma now states the targeted characterization:

\begin{lemma}\label{lem:pc} 
    A CQ \qr is parallel-correct under distribution policy \distp
    if and only if the following holds:
         \begin{itemize}
    \item[{(C1)}] for every \inducedminimal valuation $V$ for \qr,
      $$\bigcap_{\fc \in V(\body{\qr})}\distp(\fc) \ne \emptyset.$$
    \end{itemize}

\end{lemma}

\begin{proof}[Proof (sketch)]
{(if) Assume (C1) holds. Because of monotonicity, we only need to show that
$\qr(I) \subseteq \bigcup_{\node\in\nw}{\qr(\dist I(\node))}$ for
every instance $I$. To this end, let $\fc$ be a fact that is derived
by some valuation $V$ for $\qr$ over $I$. Then, there is also a minimal valuation $V'$ that is satisfying on $I$ and which derives $\fc$. 
Because of (C1), there is a node $\kappa$ where all facts required
for $V'$ meet. Hence, $\fc\in \bigcup_{\node\in\nw}{\qr(\dist I(\node))}$.

    (only-if) Proof by contraposition.} Suppose that there is a \inducedminimal valuation $V'$ for \qr for which the required facts do not meet under \distp. Consider $V'(\body{\qr})$ as input instance. Then, by definition of \inducedminimal{}ity, there is no valuation that agrees on the head-variables and is satisfied on one of the chunks of $V'(\body{\qr})$ under \distp. So, $\qr$ is not parallel-correct.
\end{proof}

\begin{example} 
        For a simple example of a \inducedminimal valuation and a non-\inducedminimal valuation, consider the CQ $\qr$, 
    \begin{align*}
        T(x,z) \leftarrow R(x,y),R(y,z),R(x,x).
    \end{align*}
    Both 
    $V = \{x\mapsto a,y\mapsto b, z\mapsto a\}$ and $V' = \{x \mapsto a, y \mapsto a,z \mapsto a\}$ are valuations for \qr.
    Notice that both valuations agree on the head-variables of \qr, but they require different sets of facts. In particular, for $V$ to be satisfying on $I$, instance $I$ must contain the facts $R(a,b), R(b,a)$, and $R(a,a)$, while $V'$ only requires $I$ to contain $R(a,a)$. 
    This observation implies that $V$ is not \inducedminimal for $\qr$.
    Further, as $V'$ requires only one fact for $\qr$, $V'$ must be \inducedminimal for $\qr$.
    
    We next argue that (C0) is not a necessary condition for parallel-correctness. Indeed, take
    $\nw=\{1,2\}$ and $\distp$ as the distribution policy mapping every fact except
    $R(a,b)$ onto node $1$ and every fact except $R(b,a)$ onto node $2$.
    Consider the valuations $V$ and $W=\{x\mapsto b, y\mapsto a, z\mapsto b\}$. Then, $R(a,b)$ and $R(b,a)$ do not meet under $\distp$, thus violating condition (C0).
    It remains to argue that $\qr$ is parallel-correct under $\distp$. 
    For every minimal valuation $U$, either $\bigcap_{\fc \in U(\body{\qr})}\distp(\fc) \ne \emptyset$ or $U$ requires both $R(a,b)$ and $R(b,a)$. But in the latter case $U$ is either valuation $V$ or $W$ as defined above which are not minimal. Thus, by Lemma~\ref{lem:pc}, query~$\qr$ is parallel-correct under~$\distp$. \qed

    \label{ex:minimal_valuations}
\end{example}

Unfortunately, condition (C1) is complexity-wise more involved than
(C0) as minimality of $V$ needs to be tested. The lower bound in
Theorem \ref{theo:subinstances-completeness} below indicates that this
can, in  a sense, not be avoided.  

Towards an upper bound for the complexity of parallel-correctness, we first discuss how minimality of a valuation can be tested. Obviously, this notion is related to the (classical) notion of minimality for conjunctive queries, as we will make precise next.  First, recall that a CQ $\qr$ is {\em minimal} if there is no equivalent CQ with strictly less atoms.

\begin{lemma}\label{lem:minimal}
Let \qr be a conjunctive query. For every injective valuation $V$ for
\qr, it holds that $V$ is minimal if and only if $\qr$ is
minimal.\footnote{We note that in a previous version of this paper,
  this lemma had a second statement (basically an explicit reduction from
  valuation minimality to CQ minimality), which turned out to be
  wrong. However, this second statement is not essential for the
  other results of the paper. Furthermore, from Proposition
  \ref{prop:minimal-complexity} it follows that a polynomial reduction
  from 
  valuation minimality to CQ minimality indeed exists.}
\end{lemma}

Lemma \ref{lem:minimal} immediately yields the following complexity
result.
\begin{proposition}\label{prop:minimal-complexity}
Deciding whether
a valuation $V$ for a CQ $\qr$ is \inducedminimal is \coNP-complete.
\end{proposition}
\begin{proof}[Proof (sketch)]
Lemma \ref{lem:minimal} yields a reduction from minimality of conjunctive queries
to minimality of valuations. Therefore,  \coNP-hardness follows from the \coNP-hardness of
minimality for CQs, which follows from \cite{HellN92}. The upper bound
is immediate from the definition of minimality of valuations and from
the fact that, for given $V_1,V_2,\qr$,
it can be tested in polynomial time whether $V_1<_\qr V_2$ holds.
\end{proof}

Now, we are ready to settle the complexity of parallel-correctness for
general conjunctive queries for a large class of
distributions. We study two settings,  $\classpfin$, where distribution policies
are explicitly enumerated as part of the input, and $\classpnrel$,
where the distribution policy is given by a black box procedure which
answers questions of the form ``$\node\in\distp(\fc)$?'' in \NP. In the latter case, the distribution is not part
of the (normal) input and therefore does not contribute to the input
size. Instead, the input has an additional parameter $n$ which
bounds the length of addresses in the considered networks. 

By $\dom_n$
we denote the set of all elements of $\dom$ that can be encoded by
strings of length at most $n$. For a distribution policy \distp (coming with a
network $\cal N$) and a
number $n$, we denote  by $\distp_n$ the distribution policy that is
obtained from \distp by (1) only distributing facts over $\dom_n$ and
(2) only distributing facts to nodes whose addresses are of length at
most $n$.

We study the following algorithmic problems for explicitly given
database instances:

\probdefbodynew{$\PCI(\classpfin)$}{CQ~$\qr$, instance $I$, and $\distp \in \classpfin$}{Is $\qr$ parallel-correct on $I$ under $\distp$ ?}

\probdefbodybb{$\PCI(\classpnrel)$}{CQ~$\qr$, instance $I$,  a natural number $n$ in unary representation}{$\distp
  \in \classpnrel$}{Is $\qr$ parallel-correct on $I$ under $\distp_n$?}

We also study the parallel correctness problem without
reference to a given database instance.

 \probdefbodynew{$\PC(\classpfin)$}{CQ~$\qr$, $\distp \in \classpfin$}{Is
   $\qr$ parallel-correct on $I$ under $\distp$, for all instances $I\subseteq
   \facts{\distp}$?}

Here, $\facts{\distp}$ denotes the set of facts \fc with $\distp(\fc)\not=\emptyset$.

\probdefbodybb{$\PC(\classpnrel)$}{CQ~$\qr$, a natural number $n$ in
  unary representation}
{$\distp \in \classpnrel$}
{Is $\qr$  parallel-correct on $I$ under $\distp_n$, for all instances $I\subseteq
   \facts{\distp_n}$?}

\begin{theorem} \label{theo:subinstances-completeness} 
\topsep=0pt
\begin{enumerate}[(a)]\itemsep=0pt
	\item $\PC(\classpfin)$ and $\PCI(\classpfin)$ are
          $\phtwo$-complete.
        \item 	$\PC(\classpnrel)$ and  $\PCI(\classpnrel)$ are
          in $\phtwo$.
\end{enumerate}
\end{theorem}
\begin{proof}[Proof (sketch)]
The upper bounds follow quite directly from Definition~\ref{def:pc-inst}, or
Lemma~\ref{lem:pc} and Proposition~\ref{prop:minimal-complexity}, respectively.
The lower bound is by a reduction from the $\phtwo$-complete
$\twoTQBF$-problem and uses a distribution policy
with only two nodes. A proof is given in Appendix~\ref{sec:Proof-SubinstancesCompleteness}.
\end{proof}
Of course, the upper bounds in Theorem
\ref{theo:subinstances-completeness} also hold if questions of the
form ``$\node\in\distp(\fc)$?'' are answered in polynomial time, or if
\distp is just given as 
a polynomial time function.

Due to the implicit representation of distributions, we cannot
formally claim $\Pi^p_2$-hardness for distribution policies from
$\classpnrel$. However, in an informal
sense, they are, of course, at least as difficult as for $\classpfin$.


%% file: transferability.tex
\section{Transferability}
\label{sec:trans}

While parallel-correctness provides a direct one-round evaluation algorithm,
it still requires a reshuffling of the data for every query. It therefore makes sense, in the context of multiple query evaluation, to consider scenarios in which such reshuffling can be avoided. To this end, we introduce 
the notion of parallel-correctness transfer which ensures that a subsequent query $\qr'$ can always be evaluated over a distribution for which a query $Q$
is parallel-correct:

\begin{definition}
For two queries $\qr$ and $\qr'$ over the same input and output schema, \emph{parallel-correctness transfers from $\qr$ to $\qr'$} when $\qr'$ is parallel-correct under every distribution policy for which $\qr$ is parallel-correct.
\end{definition}

As for parallel-correctness we first give a semantical
characterization before we study the complexity of parallel-correctness transfer. The proof of the following lemma is given in Appendix~\ref{appendix:sem:trans}.

\begin{lemma}\label{lem:trans-sem-char}
    Parallel-correctness transfers from a conjunctive query $\qr$ to a
    conjunctive query $\qr'$
    if and only if the following holds:
    \begin{itemize}
    \item[{(C2)}] for every \inducedminimal valuation $V'$ for $\qr'$,
    there is a \inducedminimal valuation $V$ for $\qr$ such that $V'(\body{\qr'}) \subseteq V(\body{\qr})$.
    \end{itemize}
\end{lemma}
\begin{proof}[Proof (sketch)]
 (if) Let $\distp$ be a distribution policy under
  which $\qr$ is parallel-correct. Let $I$ be an instance and $\fc$ a
  fact in $\qr'(I)$. It suffices to show that there is some valuation
  $V'$ for $\qr'$ that produces $\fc$ at some node. Let $V'$ be any
  minimal valuation that yields \fc and let $V$ be the minimal
  valuation of \qr, guaranteed by (C2). Since \qr is parallel-correct
  under \distp and $V$ is minimal for \qr, it follows by (C1) that all facts
  from $V(\body{\qr})$ must meet at some node \node. Since $V'(\body{\qr'})
  \subseteq V(\body{\qr})$, all facts from $V'(\body{\qr'})$ are
  available  at node \node and thus \fc is produced.

 (only-if)  Towards a contradiction, we assume that there is a \inducedminimal valuation $V'$
    for $\qr'$, for which there is no valuation $V$ for $\qr$, where
    $V'(\body{\qr'})\subseteq V(\body{\qr})$. Let $m =
    |V'(\body{\qr'})|$. In this sketch we only consider the case $m\ge 2$.
Let $I\mydef V'(\body{\qr'}) = \{\fc_1, \ldots, \fc_m\}$, $\nw \mydef \{\node_1,\ldots,
\node_m\}$, and let \distp be the distribution policy defined by
\begin{itemize}
\item $\distp(\fcB) \mydef \nw$, for every $\fcB\in\facts{\sch}\setminus I$; and
\item $\distp(\fc_i) \mydef \nw \setminus \{\kappa_i\}$, for every $i$.
\end{itemize}
Intuitively, on every instance $J$, either the facts in $J$ meet on some node under $\distp$, or $I \subseteq J$. By assumption, none of the \inducedminimal valuations for \qr requires all the facts in $I$, implying that \qr is parallel-correct under \distp.
Nevertheless, on instance $I$ under $\distp$, none of the nodes
receives all the facts in $I$, and  (by \inducedminimal{}ity of $V'$)
there is no valuation that can derive $V'(\head{\qr'})$ for a strict subset of the facts in $I$. So, $\qr'$ is not parallel-correct under \distp, which leads to the desired contradiction.
\end{proof}
The above characterization allows us to pinpoint the complexity of parallel-correctness transfers. In particular,
 we consider the following problem:

 \probdefbodynew{\PCTRANS}{CQs $\qr$ and $\qr'$}{Does parallel-correctness transfer from $\qr$ to $\qr'$?}

\begin{theorem}
	\label{theo:trans-completeness}
    \PCTRANS is \polyh{3}-complete.
\end{theorem}
\begin{proof}[Proof (idea)]
  The upper bound again follows quite directly from the semantical
  characterization (Lemma \ref{lem:trans-sem-char}) and
Proposition \ref{prop:minimal-complexity}.  The lower bound is by a
reduction from the $\phthree$-complete
	$\threeTQBF$-problem. 	Let thus $\phi= 
		\forall \vx
		\exists \vy
		\forall \vz\,
			\psi(\vx,\vy,\vz)$
be a formula
	with a quantifier-free  propositional formula $\psi$ in 3-DNF over
	variables $\vx=(\seqx)$, $\vy=(\seqy)$, and
        $\vz=(\seqz)$. From $\phi$ we
        construct a CQ $\qr'$ with head $H(\seqx,\vart,\varf)$ and
        a CQ $\qr$ with head
        $H(\seqx,\seqy,\vart,\varf)$.
        We then show, in a nutshell, that falseness of $\phi$
        corresponds to the existence of a valuation $V'$ for  $\qr'$
        (inducing a truth assignment $\beta_\vx$ to $\tup x$), such
        that for every valuation $V$ for \qr (with induced
        truth assignment  $\beta_\vy$), with $V'(\body{\qr'}) \subseteq V(\body\qr)$, from any truth assignment
        $\beta_\vz$ with $(\beta_\vx \cup \beta_\vy \cup \beta_\vz) \not\models
             \psi$ a valuation $V^\ast$ is obtained with $V^\ast<_\qr
             V$, implying that $V$ is not minimal. As such a truth assignment
        $\beta_\vz$ exists (after the choice
        of $\beta_\vx$ and for arbitrary $\beta_\vy$) if and only if $\phi$ is  false, it can be
        shown that        (C3) holds if and only if $\phi$ is true.

A complete proof is given in Appendix~\ref{appendix:trans-compl}.
\end{proof}

It is an easy observation that, if we require each valuation of
$\qr$ to be minimal, then condition (C2) yields a better, \polyh{2},
complexity bound. Surprisingly, in this case, we even get a complexity
drop to \NP, as will be shown in Theorem \ref{theo:transfer-strong-complexity} below.  We next introduce the notions needed for this result.
\begin{definition}
  A conjunctive query $\qr$ is \emph{strongly minimal} if all its
  valuations are minimal.
\end{definition}

We give some examples illustrating this definition. In Lemma~\ref{lem:cond_st_minimal}, we present a sufficient condition for CQs to be strongly minimal.

\begin{example}
    For an example of a strongly minimal query, consider query $\qr_1$, 
    \begin{align*}
        T(x_1, x_2, x_2, x_4) \leftarrow R(x_1, x_2), R(x_2, x_3), R(x_3, x_4).
    \end{align*}
    Notice that, by fullness of $\qr_1$, there are no two distinct valuations for $\qr_1$
    that derive the same fact. Hence, every valuation of $\qr_1$ must indeed
    be minimal.

    For another example, consider the query $\qr_2$,
    \begin{align*}
        T() \leftarrow R_1(x_1, x_2), R_2(x_2, x_3), R_3(x_3, x_4).
    \end{align*}
    As each atom in the body of $\qr_2$ has a different relation symbol,
    each valuation of $\qr_2$ yields exactly three different facts and
    therefore, each valuation is minimal.\qed
\end{example}
It is easy to see that every strongly minimal CQ is also a minimal CQ, but the
converse is not true as witnessed by the query of
Example~\ref{ex:minimal_valuations}, which is minimal but not strongly minimal.

The following lemma now provides a characterization of parallel-correctness
transfer for strongly minimal queries.

    \begin{lemma}\label{lem:transfer-strong}
    Let $\qr'$ be a CQ and let $\qr$ be a strongly minimal CQ.
    Parallel-correctness transfers from $\qr$ to $\qr'$
    if and only if the following holds:
    \begin{itemize}
    \item[{(C3)}] there is a simplification $\theta$ for $\qr'$ and a substitution $\rho$ for $\qr$
        such that $\body{\theta(\qr')} \subseteq \body{\rho(\qr)}$.
    \end{itemize}
\end{lemma}

 \begin{proof}
 
 We show that, for strongly minimal \qr, (C2) and (C3) are equivalent.

We first show that (C3) implies (C2). It suffices to show that if (C3)
 holds then for every minimal valuation $V'$ for $\qr'$, there is a
 valuation $V$ for \qr such that $V'(\body{\qr'}) \subseteq
 V(\body{\qr})$. By strong minimality of $Q$, we can then conclude that
 $V$ is actually minimal. 

Let $V'$ be a minimal valuation for $\qr'$ and let $\simp$
and $\rho$ be as in (C3). As $\simp$ is a simplification,
$\head{\simp(\qr')}=\head{\qr'}$ and
$\body{\simp(\qr')}\subseteq\body{\qr'}$. Therefore $(V'\circ\simp)$ is
also a valuation for $\qr'$ with $(V'\circ\simp)(\body{\qr'})=
V'(\body{\simp(\qr')})\subseteq V'(\body{\qr'})$ and by minimality of
$V'$ the latter inclusion is actually an equality. 

By (C3), $\body{\simp(\qr')} \subseteq \body{\rho(\qr)}$, therefore
$V'$ is a partial valuation for $\rho(\qr)$. Let $V''$ be some
arbitrarily chosen extension of $V'$ that is a (total) valuation for
$\rho(\qr)$.
Then,\\
$V'(\body{\qr'}) = V'(\body{\simp(\qr')}) =
V''(\body{\theta(\qr')})$\mynewr[2mm] 
$\subseteq
V''(\body{\rho(\qr)}) =(V''\circ\rho)(\body\qr).$\\

Thus, $V\mydef V''\circ\rho$ is the desired valuation for \qr.\\

We next show that (C2) implies (C3). Actually, this implication even
holds without the assumption that \qr is strongly minimal.
Let us therefore assume that (C2) holds. 
We choose $\theta$ as an arbitrary simplification that minimizes
$\qr'$.  Such a simplification can be found thanks to
\cite{DBLP:conf/stoc/ChandraM77}. In particular, $\theta(\qr')$ is a minimal CQ that is
equivalent to $\qr'$. 

Let $V'$ be an injective valuation for $\qr'$. We claim that $V'\circ
\theta$ is  a minimal valuation for $\qr'$. Towards a contradiction, let us assume
that there is a valuation $V''$ such that $V''<_{\qr'} V'\circ
\theta$. Since $\theta$ is the identity on the head variables, $V'$
is injective, and $V'$ and $V''$ agree on $\head{\qr}$, we can
conclude that $((V')^{-1}\circ
V'')(\head{\qr'})=\head{\qr'}$, thus $(V')^{-1}\circ
V''$ is a homomorphism from $\qr$ to $((V')^{-1}\circ
V'')(\qr)$. Furthermore, $((V')^{-1}\circ
V'')(\body{\qr})\subseteq \body{\theta(\qr)}\subseteq\body{\qr}$,
therefore the identity is a homomorphism from $((V')^{-1}\circ
V'')(\qr)$ to $\qr$. Together,
$((V')^{-1}\circ
V'')(\qr)$ is equivalent to $\qr$. Furthermore, $((V')^{-1}\circ
V'')(\body{\qr})\subsetneq
(V')^{-1}(V'(\body{\qr}))=\body{\theta(\qr)}$, contradicting the
minimality of $\theta$. We thus conclude that $V'\circ
\theta$ is indeed a minimal valuation for $\qr'$.

By (C2), there exists a minimal valuation
$V$ for $\qr$ such that $ V'(\body{\theta(\qr')}) = (V'\circ \theta)(\body{\qr'})  \subseteq
V(\body{\qr})$. 
Now, let $f$ be an extension of $(V')^{-1}$, which maps values that
occur in $V(\body{\qr})$ but not in $V'(\body{\qr'})$ in an arbitrary
fashion and let $\rho\mydef (f\circ V)$. 
Then, 
\begin{multline*}
    \body{\theta(\qr')} =  (V')^{-1}(V'(\body{\theta(\qr')}))\\
    =f(V'(\body{\theta(\qr')})) 
    = f((V' \circ \theta)(\body{\qr'})) \\
     \subseteq f(V(\body{\qr})) 
    = \rho(\body{\qr})
  = \body{\rho(\qr)}.
\end{multline*}
Thus, $\theta$ and $\rho$ witness condition (C3).
\end{proof}

\begin{theorem}\label{theo:transfer-strong-complexity}
    \PCTRANS restricted to inputs with strongly minimal  \qr is
    \NP-complete.
\end{theorem}
\begin{proof}[Proof (sketch)]
  The upper bound follows from Lemma \ref{lem:transfer-strong} by the
  observation that condition (C3) can be checked by a straighforward
  \NP-algorithm. The lower bound follows from
  Proposition~\ref{prop:cthree} below.
\end{proof}

Theorem \ref{theo:transfer-strong-complexity} assumes that it is known
that $Q$ is strongly minimal. We complete the picture by investigating
the complexity of the problem to decide whether a CQ is strongly minimal.

We first give a lemma that generalizes the above examples into a sufficient (but not
necessary) condition for strong minimality.

 In particular, Lemma~\ref{lem:cond_st_minimal} implies that every
    full CQ and every CQ without self-joins is strongly minimal.
    We say that an atom in a CQ is a {\em self-join atom} when the relation name of that atom occurs more than once in $\qr$. For instance, in the query  $T() \leftarrow R(x_1,x_2), R(x_2, x_1)$ both $R(x_1,x_2)$ and $R(x_2,x_1)$ are self-join atoms.

\begin{lemma}\label{lem:cond_st_minimal} 
Let $\qr$ be a CQ. Then $\qr$ is
strongly minimal when the following condition holds: 
if a variable $x$ occurs at a position $i$ in some self-join atom
and not in the head of $\qr$, then all self-join atoms have $x$ at position $i$.
\end{lemma}
\begin{proof}[Proof (sketch)]
The proof is by contraposition, i.e., we show that if there is a valuation for $\qr$ which is not minimal then the condition is not satisfied.
To this end, let $V$ and $V'$ be valuations for $\qr$ which agree on the head-variables and where $V'(\body{\qr}) \subsetneq V(\body{\qr})$. 

Then, there are at least two atoms $A_1=R(x_1,\dots,x_k)$
and $A_2=R(y_1,\dots,y_k)$ in the body of $\qr$ that collapse under $V'$, but not under $V$. 
That is, $V'(A_1)=V'(A_2)$ and $V(A_1)\neq V(A_2)$.
So, under $V'$ all the variables in $A_1$ and $A_2$ on matching positions must be mapped on the same constant, $V'(x_i)=V'(y_i)$ for each $i \in \{1,\dots,k\}$, while for $V$ there is a position~$j \in \{1,\dots,k\}$ where this is not the case, $V(x_{j}) \neq V(y_{j})$. Obviously, at
least one of these variables must then be a non-head variable. So, either only $x_{j}$ is a head variable, or only $y_{j}$ is a head
variable, or both are distinct non-head variables. In both cases the condition is not
satisfied.
\end{proof}

\begin{example} For an example of a strongly minimal CQ that does not satisfy the
    condition in Lemma~\ref{lem:cond_st_minimal}, consider
    query $\qr_3$, 
    \begin{align*}
        T() \leftarrow R(x_1,x_2), R(x_2, x_1).
    \end{align*}
    Notice that $\qr_3$ is indeed strongly minimal, because every valuation for $\qr_3$
    either maps $x_1$ and $x_2$ on the same value, and thus requires only one
    {fact} 
    where both
    values are equal, or maps $x_1$ and $x_2$ onto two distinct values, and thus requires
    exactly two {facts} 
    where both
    values are distinct. 
\end{example}

Finally, we establish the complexity of deciding strong minimality.

\begin{lemma}
	\label{lem:comp_st_minimal}
    Deciding whether a CQ is strongly minimal is \coNP-complete.
\end{lemma}
\begin{proof}[Proof (sketch)]
	The complement problem is easily seen to be in $\NP$: for two guessed
	valuations $\minVal,V$ (encoded in length polynomial of the query~$\qr$) it can
	be checked in polynomial time whether $\minVal<_\qr V$.

    The lower bound is proved in Appendix~\ref{ssec:Lemma_comp_st_minimal}.\
\end{proof}



%% file: families.tex
\section{Families of Distribution Policies}
\label{sec:families}

Parallel-correctness transfer can be seen as a generalization
of parallel-correctness. In both cases, the goal is to decide whether a query
can be correctly evaluated by evaluating it locally at each node. However,
for parallel-correctness \emph{transfer}, the question whether $\qr'$ is
parallel-correct is not asked for a
particular distribution policy but for the \emph{family} of those distribution
policies, for which $\qr$ is parallel-correct.\footnote{A {family} of distribution policies is just a set of distribution
policies. }

In this section, we study the parallel-correctness problem for
other kinds of families of distribution policies that can be
associated with a given query $\qr$. In Section~\ref{sec:pc:sec5}, we will identify classes of families of policies, for which (C3) characterizes parallel-correctness.
For these classes we conclude that it is \NP-complete to decide, whether for the family
$\cal F$ of policies associated with some given CQ \qr, a
CQ $\qr'$ is parallel-correct for all distributions from $\cal F$.  
In Section~\ref{sec:declarative}, we
will see that this, in particular, holds for the
families of distribution policies related to the practical Hypercube algorithm, that was previously investigated in several works~\cite{AfratiUllman10,DBLP:conf/pods/BeameKS13,DBLP:conf/pods/BeameKS14,DBLP:conf/sigmod/GangulyST90,DBLP:journals/jlp/GangulyST92}. In fact, we even show that this holds for a more general class of distribution policies specified in a declarative formalism.

\subsection{Parallel-correctness}
\label{sec:pc:sec5}

We start with the following definition:  
\begin{definition} \label{def:family}
A query \qr is {\em parallel-correct for a family\/}
$\cal F$ of distribution policies if it is parallel-correct {under} every
{distribution} policy from $\cal F$.
\end{definition}

We call a distribution policy \distp \emph{$\qr$-generous} for a CQ $\qr$, if,
for every valuation $V$ for $Q$, there is a node $\node$ that contains
all facts from $V(\body{\qr})$. A family of distribution policies $\cal F$
is $\qr$-generous if every policy in $\cal F$ is.
For an instance $I$, a distribution policy \distp
is called \emph{$(\qr,I)$-scattered} if for each node $\node$ there is a valuation $V$ for $\qr$, such that $\dist{I}(\kappa)\subseteq V(\body{\qr})$.
We then say that a family $\cal F$ of distribution policies
is $\qr$-scattered if $\cal F$ contains a $(\qr,I)$-scattered policy
for every $I$. A $(\qr,I)$-scattered policy that is also
$\qr$-generous yields the finest possible partition of the facts of
$I$ and thus, intuitively, scatters them as much as possible.

    \begin{lemma}\label{lem:transfer-families}
Let \qr be a CQ and let
 $\cal F$ be a family of distribution policies that is
 \qr-generous and \qr-scattered.
Then for every  CQ $\qr'$, $\qr'$ is parallel correct
for  $\cal F$      if and only if:
    \begin{itemize}
    \item[{(C3)}]there is a simplification $\theta$ for $\qr'$ and a substitution $\rho$ for $\qr$
        such that $\body{\theta(\qr')} \subseteq \body{\rho(\qr)}$.
    \end{itemize}
\end{lemma}
We emphasize  that Lemma \ref{lem:transfer-families} uses the same condition (C3) as Lemma \ref{lem:transfer-strong}.

\begin{proof}[Proof (sketch)]
  (if) Let $I$ be a database
  for $\qr'$, \distp a distribution policy from $\cal F$, and let
  $\theta$ and $\rho$ be as guaranteed by (C3). We show that each fact
  from $\qr'(I)$ is produced at some node.
Let
  $V'$ be a valuation that yields some fact $\fhead\mydef
  V'(\head{\qr'})$ and let $V''$ be an arbitrary extension of $V'$
  for $\rho(\qr)$. As $\theta$ is a simplification, $(V'\circ\theta)$
  also yields the fact $\fhead$.
By (C3) we get $(V'\circ
  \theta)(\body{\qr'})=V'(\body{\theta(\qr')})\subseteq
  V''(\body{\rho(\qr)})=(V''\circ\rho)(\body{\qr})$. As \distp is $\qr$-generous, there is some node
  \node that has all facts from $(V''\circ\rho)(\body{\qr})$ and
  therefore all facts from $(V'\circ\theta)(\body{\qr'})$, and thus \fhead is
  produced at \node.

\smallskip
(only-if)
   Suppose $\qr'$ is parallel-correct under all distribution policies
   in $\cal F$. Let $V'$ be some injective valuation for
   $\qr'$. Denote $I\mydef V'(\body{\qr'})$ and $\fhead\mydef V'(\head{\qr'})$.
   Let \distp be some $(\qr,I)$-scattered distribution
   policy from $\cal F$. 
    Because $\qr'$ is parallel-correct under $\distp$, there must be a
    node \node that outputs $\fhead$ when $I$ is distributed according
    to $\distp$. Therefore, there is a valuation $W'$ for $\qr'$ such
    that \node contains all facts from
    $W'(\body{\qr'})$ and $W'(\head{\qr'})=\fhead$. 
 We claim that $\theta\mydef (V')^{-1}\circ W'$ is a simplification of
 $\qr'$. Indeed, this substitution is well-defined thanks to 
the injectivity of $V'$ and furtermore
$((V')^{-1}\circ W')(\head{\qr'})=\head{\qr'}$ and $((V')^{-1}\circ
W')(\body{\qr'})\subseteq\body{\qr'}$, as $W'(\body{\qr'})\subseteq
I=V'(\body{\qr'})$ and $(V')^{-1}$ maps $I$ back to $\body{\qr'}$.

    {As $\distp$ is $(\qr,I)$-scattered, there is a valuation $V$
      such that $\dist{I}(\kappa)\subseteq V(\body{\qr})$.}
    {Then, let $g$ be some mapping from $\img{
    V}$ to \uvar such that for all $d\in \img{W'}$,
    $g(d)=g'(d)$.} 
We define the renaming $\rho\mydef g \circ V$ and show that with these
choices, $\body{\theta(\qr')}\subseteq\body{\rho(\qr)}$, and thus (C3)
holds.

Let $R(x_1,\ldots,x_k)\in
\body{\theta(\qr')}$. Then, there is an atom $R(y_1,\ldots,y_k)\in \body{\qr'}$
with $W'(R(\bar y))\in \dist{I}(\kappa)$ and, for each $i$,  $x_i=(V')^{-1}(W'(y_i))$.
So, as $\dist{I}(\kappa)\subseteq V(\body{\qr})$,
$W'(R(\bar y))\in V(\body{\qr})$ and there is some atom $R(z_1,\ldots,z_k)\in \body{\qr}$ such that $W'(R(\bar y)) = V(R(\bar z))$. 
Clearly, $W'(y_i)=V(z_i)$ for all $i$. By definition of $g$,
it then follows that $x_i=(V')^{-1} (W'(y_i))=g(V(z_i))$ for all $i$.
Therefore, $R(x_1,\ldots,x_k)\in
\body{\rho(\qr)}$, as desired.  
\end{proof}

\begin{theorem}
It is \NP-complete to decide, for given CQs $\qr$ and $\qr'$, whether $\qr'$ is
parallel-correct for \qr-generous 
and $\qr$-scattered  families of distribution policies.
    \label{theo:families-complexity}
\end{theorem}
The proof of this theorem shows in particular, that $\qr'$ is either
parallel-correct for \emph{all} \qr-generous 
and $\qr$-scattered  families of distribution policies or for \emph{none} of them.

\begin{proof}[Proof (sketch)]
    The upper bound follows from Lemma~\ref{lem:transfer-families}
    and the fact that (C3) can be checked by an \NP-algorithm. Indeed
    such an algorithm only needs to guess $\theta$ and $\rho$ and to
    verify (in polynomial time) that $\body{\theta(\qr')} \subseteq \body{\rho(\qr)}$.

The
    lower bound follows by Lemma~\ref{lem:transfer-families} and the
    following Proposition \ref{prop:cthree}. 
\end{proof}
\begin{proposition}\label{prop:cthree}
  It is \NP-hard
    to decide, whether  for CQs \qr and $\qr'$
    condition (C3) holds. This statement remains true if either \qr
    or $\qr'$ is restricted to acyclic queries. It also remains true if both CQs are
    Boolean and if \qr is full.
\end{proposition}

\begin{remark}
The proof of Proposition \ref{prop:cthree} is given in
        Appendix~\ref{app:proof-families-complexity}. In both cases
        (\qr acyclic or $\qr'$ acyclic) it is by a reduction
    from graph $3$-colorability.  The first reduction, in which the input graph is encoded in $\qr'$ and the valid
color-assignments in $\qr$ is straightforward. As it only uses a fixed number of colors, $\qr$ can be made acyclic by adding an atom to $\qr$ that contains all allowed colors.

The second reduction, in which the graph is encoded in $\qr$ and the
valid color-assignments in $\qr'$,  is a bit more involved.

The reader may now wonder whether 
        \NP-hardness remains when both $\qr$ and $\qr'$ are required to be acyclic. 
    When relations of arbitrary arity are allowed, this is indeed the case:  acyclicity is then easily achieved by using one atom containing all variables of the query. 
        Under bounded-arity database schemas, however, the complexity of
    parallel-correctness transfer for acylic queries remains open. \qed
\end{remark}

\subsection{Hypercube Distribution Policies}
\label{sec:declarative}

\input{declarativepolicies.tex}


%% file: declarativepolicies.tex

In the following, we give a short definition of Hypercube
distributions and settle the complexity of the
parallel-correctness transfer problem for families $\cal H(\qr)$ of Hypercube
distributions for some CQ \qr with the help of the results of
Section \ref{sec:pc:sec5}. We highlight how Hypercube
distributions can be specified in a rule-based fashion, which we consider useful also for more general distributions.

Let \qr be a conjunctive query with (body) variables
$x_1,\ldots,x_k$. A collection $H=(h_1,\ldots,h_k)$ of hash
functions\footnote{A hash function is a partial mapping from $\dom$
to a finite set whose elements are sometimes referred to as buckets.}  (called a \emph{hypercube} in the following) determines a \emph{hypercube distribution} $\distp_H$ for $\qr$ in
the following way. For each $i\in\{1,\ldots,k\}$, we let $A_i\mydef
\img{h_i}$ and define the  \emph{address space} $\as$ of $\distp_H$ as
the cartesian product $A_1\times \cdots\times A_k$. 

In a nutshell, $\distp_H$ has one node per address in $\as$ and
distributes, for every valuation $V$ of $\qr$, every fact $\fc=V(A)$, where $A$ is an atom of
$\qr$, to all nodes whose address $(a_1,\ldots,a_k)$
satisfies $a_i=h_i(V(x_i))$, for all variables $x_i$ occurring in
$A$.

For the declarative specification of $\distp_H$ we make use of 
predicates\footnote{For the purpose of specification it is irrelevant
  whether these predicates are materialized in the database.} $\sbucket i$ and $\abucket i$, where
$\sbucket i(a,b)$ holds, if $h_i(a)=b$, and  $\abucket i(b)$
holds, if $b\in \img{h_i}$.

With these predicates, $\distp_H$ can be specified by stating, for each
atom $R(y_1,\ldots,y_m)$ of \qr, one rule
\begin{align*}
    T_R(z_1,\ldots,z_k; y_1,\ldots,y_m) \gets &
        R(y_1,\ldots,y_m), \\
        & \atomB_1, \ldots, \atomB_k.
 \end{align*}
Here, for each  $i\in\{1,\ldots,k\}$, $\atomB_i$ is $\sbucket i(x_i,z_i)$,
if $x_i$ occurs in $y_1,\ldots,y_m$, and $\atomB_i$ is $\abucket i(z_i)$, otherwise.

The semantics of such a rule is straightforward. For each valuation
$V$ of the variables $z_1,\ldots,z_k,x_1,\ldots,x_k$, that makes the body of the rule true, the fact
$R(V(y_1),\ldots,V(y_m))$ is distributed to the node with address
$(V(z_1),\ldots,V(z_k))$. We emphasize that the variables
$y_1,\ldots,y_m$ need not be pairwise distinct and that  $\{y_1,\ldots,y_m\}\subseteq\{x_1,\ldots,x_k\}$.

\begin{remark}
  It is evident that one could use more general rules to specify
  distribution policies. More than one atom with a database relation
  could be in the body, and there could be other additional predicates
  than those derived from hashing functions. Furthermore, the address
  space could be defined differently. \qed
\end{remark}

For a CQ $\qr$, we denote by ${\cal H}_\qr$ the family 
of distribution policies $\{\distp_H\mid H\text{ is a hypercube for } \qr\}$.

\begin{lemma} \label{lem:hc:generous-scattered}
Let $\qr$ be a CQ. Then ${\cal H}_\qr$ is $\qr$-generous and 
$\qr$-scattered.
\end{lemma}
\begin{proof}
Let $Q$ be a CQ with $\vars{\qr}=\{u_1,\ldots,u_k\}$.

We first show that every policy $\distp_H\in {\cal H}_\qr$ is
$\qr$-generous. To this end, let $H$ be a hypercube and
let $V$ be a valuation for $\qr$. Then, by definition, for the node $\kappa$ with address $(h_1(V(u_1)),\ldots,h_k(V(u_k)))$,  $\kappa \in {\distp_H(\fc)}$ 
for every $\fc\in V(\body{\qr})$.

We now show that ${\cal H}_\qr$ is $\qr$-scattered. Thereto, let $I$ be an instance. For every $i\leq k$, we choose $A_i\mydef\adom
I$
and let $h_i(a)\mydef a$, for every $a\in A_i$. 
Let $\kappa$ be an arbitrary node and let $(a_1,\ldots,a_k)$  be its
address. Let $V$ be
the valuation mapping $u_i$ to $a_i$, for each $i$. 
 Let $R(d_1,\ldots,d_m)\in
\mathit{dist}_{\distp_H}(I)(\kappa)$ thanks to some rule 
\begin{align*}
    T_R(z_1,\ldots,z_k; y_1,\ldots,y_m) \gets &
        R(y_1,\ldots,y_m), \\
        & \atomB_1, \ldots, \atomB_k.
 \end{align*}
By definition of the hash functions, every valuation that satisfies the body of this rule, maps $x_i$ to
$a_i$, for every $x_i$ that appears in $R(y_1,\ldots,y_m)$. However,
as this valuation coincides with $V$ on  $y_1,\ldots,y_m$, it maps
$R(y_1,\ldots,y_m)$ to an element of $V(\body{\qr})$. 
Therefore,  $\mathit{dist}_{\distp_H}(I)(\kappa) \subseteq V(\body{\qr})$. 
\end{proof}

\begin{corollary}\label{coro:hyper}
It is \NP-complete to decide, for  given conjunctive queries $\qr,\qr'$,
whether $\qr'$ is parallel-correct for ${\cal H}_\qr$.
\end{corollary}

\begin{remark}
  It is easy to see that Lemma \ref{lem:hc:generous-scattered} and
  then the upper bound of Corollary \ref{coro:hyper} holds for more
  general families of distribution policies. As an example, one could
  add further atoms of \qr as ``filters'' to the bodies of the above
  rules.\qed
\end{remark}


%% file: conclusions.tex

\section{Conclusions}
\label{sec:conclusion}

We have introduced parallel-correctness as a framework for studying one-round evaluation algorithms for the evaluation of queries under arbitrary distribution policies. We have obtained tight bounds on the complexity of deciding parallel-correctness and the transferability problem for conjunctive queries. For general conjunctive queries, these complexities reside in different levels of the polynomial hierarchy (even when considering {Hypercube} distributions). Since  
the considered problems are static analysis problems that relate to queries and not to instances (at least in the case of transferability), such complexities do not necessarily put a burden on practical applicability. Still, it would be interesting to identify fragments of conjunctive queries or particular classes of distribution policies that could render these problems tractable. In addition, it would be interesting to explore more expressive classes of queries and other families of distribution policies. 

The notion of parallel-correctness is directly inspired by {Hypercube}
where the result of the query is obtained by aggregating (through union) the evaluation of the original query over the distributed instance. Other possibilities are to consider more complex aggregator functions than union and to allow for a different query than the original one to be executed at computing nodes.


%% file: appendix.tex

\newpage
\onecolumn
\section*{Appendix}
\label{sec:appendix}

\input{appendix-preliminaries.tex}
\input{appendix-parallelcorrectness.tex}
\input{appendix-transferability.tex}
\input{hc-complex-appendix}


%% file: appendix-preliminaries.tex
\section{Preliminaries}
We recapitulate some classical definitions and facts regarding the polynomial
hierarchy, which will be used in the following proofs for lower and upper
bounds.
We are only concerned with classes $\coNP=\polyh{1}$, $\phtwo$, and $\phthree$
of the first, second, and third level of the hierarchy as well as with their
complement classes $\NP=\copolyh{1}$, $\cophtwo$, and $\cophthree$.

\begin{definition}
	\label{def:phthree}
	Class $\phthree$ consists of those problems~$L$ over an
	alphabet~$\Sigma$, for which there is an algorithm~$\alg$ and a
	polynomial~$p$ such that a string $w \in \Sigma^\ast$ is in $L$ if and only if for
	every $x \in \Sigma^\ast$ there exists $y \in \Sigma^\ast$ such that for
	every $z \in \Sigma^\ast$ algorithm~$\alg$ accepts $(w,x,y,z)$ and its running
	time is bounded by $p(|w|)$. 	

Similarly, $\phtwo$ consists of those $L$, for which there is an algorithm~$\alg$ and a
	polynomial~$p$ such that a string $w \in \Sigma^\ast$ is in $L$ if and only if for
	every $x\in \Sigma^\ast$ there exists $y \in \Sigma^\ast$ such that algorithm~$\alg$ accepts $(w,x,y)$ and its running
	time is bounded by $p(|w|)$.
\end{definition}

\begin{remark}
	When we show that an algorithmic problem belongs to one of
        these classes, we do not apply the definition
        literally. Instead, as usual, we deal with 
        mathematical objects like relations and valuations directly
        and assume that their encoding as strings is done in a
        straightforward way.
\end{remark}

The following problem for quantified boolean \formulae (QBF) is complete
class~$\phthree$. The quantifier structure of the input resembles that for
inputs $x,y,z$ in the acceptance condition for the algorithm in
Definition~\ref{def:phthree}. 

\medskip
\noindent
\problemdefinition
	{\pithreeTQBF}
	{Formula $\phi=\forall\vx\exists\vy\forall\vz\; \psi(\vx,\vy,\vz)$ where $\psi$
	is a propositional formula} {
		Does for every truth assignment $\beta_\vx$ on $\vx$
		exist a truth assignment $\beta_\vy$ on $\vy$ such that
		for every truth assignment $\beta_\vz$ on $\vz$, truth
		assignment $\beta\mydef(\beta_\vx \cup \beta_\vy \cup \beta_\vz)$ satisfies
		$\psi$, $\beta\models\psi$?
	}
	
\medskip
\noindent
Problem~$\twoTQBF$ is defined analogously with formulas of the form $\phi=\forall\vx\exists\vy\; \psi(\vx,\vy)$.
In general, the number~$i$ of quantifier blocks in an input formula for $\pivarTQBF$ corresponds with level~$i$ of class $\polyh{i}$ in the polynomial hierarchy.
Also, the complement classes
$\cophtwo=\coClass{\phtwo}$ and $\cophthree=\coClass{\phthree}$ are defined like
$\phtwo$ and $\phthree$, respectively, but every quantifier is replaced by its
dual.

\begin{remark}
	\label{rem:pik-completeness-nf}
	Problems $\twoTQBF$ and $\threeTQBF$ are well-known to be complete for $\phtwo$
	and $\phthree$, respectively. They remain hard if $\psi$ is
	assumed to be in 3-CNF (for $\twoTQBF$) or in 3-DNF (for $\threeTQBF$)
	\cite[Theorem~4.1.2]{DBLP:journals/tcs/Stockmeyer76}.
\end{remark}

Below we often use the notation $\langle o \rangle$ for an object $o$ to denote some reasonable string encoding of $o$ over some finite alphabet.


%% file: appendix-parallelcorrectness.tex
\section{Proofs for Section \ref{sec:parallel-correctness}: Parallel-Correctness}


\subsection{Proof of Lemma~\ref{lem:minimal}}

In the following let $\qr$ be a CQ.  

We show that there is a non-minimal injective valuation
     $V$ for $\qr$ if and only if $\qr$ is not minimal.

    (if) Suppose that $\qr$ is not minimal.  Then, by
    \cite{DBLP:conf/stoc/ChandraM77} there is a folding $h$ for $\qr$, where
    $\body{h(\qr)} \subsetneq \body{\qr}$ and $\head{h(\qr)} = \head{\qr}$. 
        Let $V$ be an arbitrary injective valuation for $\qr$.
    Injectivity implies that $|V(\body{\qr})|=|\body{\qr}|$, that is
    the number of facts in $V(\body{\qr})$ equals the number of atoms
    in $|\body{\qr}|$.

Since $h(\qr)$ only has variables that also appear in \qr, $V$
    is a valuation for $h(\qr)$ as well. However, thanks to
    $\body{h(\qr)} \subsetneq \body{\qr}$, $h(\body{\qr})$ has fewer
    atoms than $\body{\qr}$, therefore
    valuation $|(V\circ h)(\body{\qr})|$ has fewer facts than
    $V(\body{\qr})$. Thus, $(V\circ h)$ is a counterexample for the
    minimality of $V$, since $(V\circ h)(\body{\qr}) = V(\body{h(\qr)})\subseteq V(\body{\qr})$ and
    $(V\circ h)(\head{\qr}) = V(\head{h(\qr)}) = V(\head{\qr})$. 

(only-if) Suppose there is an injective valuation $V$ for $\qr$
and a valuation $V'$ for $\qr$, such that $V'<_\qr V$. Then,
$h\mydef(V^{-1}\circ V')$ is a homomorphism from \qr to itself, as $\body{h(\qr)} = V^{-1}(V'(\body{\qr})) \subsetneq
V^{-1}(V(\body{\qr})) = \body{\qr}$, and $\head{h(\qr)} = V^{-1}(V'(\head{\qr})) =
V^{-1}(V(\head{\qr})) = \head{\qr}$. 
Therefore $h(\qr)$ equivalent to \qr, thanks to the homomorphism theorem (see, e.g., \cite{ahv_book}).

\subsection{Proof of Theorem~\ref{theo:subinstances-completeness}}
\label{sec:Proof-SubinstancesCompleteness}

\newmathabb{\addInput}{x}
\newmathabb{\addInputk}{x_1 \circ\dots\circ x_k}
\newmathabb{\addInputi}{x_i}
\newmathabb{\algNPDist}{\mathcal{A}_{\distp}}
	
\noindent
The upper bounds in Theorem~\ref{theo:subinstances-completeness}(b)
follow from Proposition~\ref{prop:pci-nrel-in-phtwo} for $\PCI(\classpnrel)$ and
Proposition~\ref{prop:pc-nrel-in-phtwo} for $\PC(\classpnrel)$. They also
imply the upper bounds of Theorem~\ref{theo:subinstances-completeness}~(a), as
stated in Corollary~\ref{cor:pci-fin-in-phtwo} for $\PCI(\classpfin)$ and
Corollary~\ref{cor:pc-fin-in-phtwo} for $\PC(\classpfin)$. The matching lower
bound for $\PCI(\classpfin)$ is stated in
Proposition~\ref{prop:pci-fin-phtwo-hard} and the matching lower bound for
$\PC(\classpfin)$ in Proposition~\ref{cor:pc-fin-phtwo-hard}.

\subsubsection{Upper bounds}

Before we prove the upper bounds, we discuss how to use distribution policies
from $\classpnrel$.

\medskip
\noindent
A distribution policy $\distp \in \classpnrel$ is an $\NP$-testable relation.
This means there exists a (deterministic) algorithm~$\algNPDist$ with time bound a polynomial in $\langle \node,\fc \rangle$ that accepts input $(\langle
\node,\fc\rangle, \addInput)$ for some string~$\addInput$ if and only if
$\node \in \distp(\fc)$. 
We use algorithm~$\algNPDist$ as a subroutine in the following algorithms, as
described below.

\begin{remark}[Use of subroutine~$\algNPDist$]
	\label{rem:check-assign}
	Let $V$ be a valuation for a query~$\qr$ with $k$ body atoms and let
	$\node$ be some node. We assume some additional input
	string~$\addInput=\addInputk$, where each substring $\addInputi$ has a
	 length polynomial in $V(\body{\qr})$ and the representation size of $\node$.
		An algorithm can \enquote{test} (w.r.t.\ $\addInput$) whether there is a fact in
	$V(\body{\qr}) = \{\fc_1,\dots,\fc_\ell\}$ that is \emph{not} assigned to
	node~$\node$ under distribution policy~$\distp$, where $\ell \le k$. To this
	end, the algorithm invokes $\algNPDist$ as a subroutine with inputs $(\langle
	\node,\fc_i\rangle, \addInputi)$ for each $i \in \{1,\dots,\ell\}$.
	If any input is rejected, the algorithm accepts, otherwise it rejects.
		The running time is obviously bounded by the size of
	$V(\body{\qr})$ and the represention size of $\node$.
\end{remark}

In the following proofs, the representation size of nodes is explicitly or
implicitly bounded by the length of the input.

\begin{proposition}
	\label{prop:pci-nrel-in-phtwo}
	Problem $\PCI(\classpnrel)$ is in $\phtwo$.
\end{proposition}

\begin{proof}
	We show that the complement problem $\overline{\PCI(\classpnrel)}$ is in
	$\cophtwo=\coClass{\phtwo}$.

	By Definition~\ref{def:pc-inst}, a query~$\qr$ is not parallel-correct on
	instance~$I$ under distribution policy~$\distp_n$ if and only if there is a
	fact in $\qr(I)$ that cannot be derived on any node. This condition is equivalent to
	the existence of a valuation $V$ that satisfies $\qr$ on $I$ such that for
	every valuation $V'$ that satisfies $\qr$ on $I$ and derives the same
	fact, $V(\head{\qr}) = V'(\head{\qr})$, no node in $\dom_n$ is assigned all
	facts from $V'(\body{\qr})$.
	
	This condition, in turn, can be translated straightforwardly into a
	$\cophtwo$-algorithm for $\overline{\PCI(\classpnrel)}$. Valuations that
	satisfy $\qr$ on $I$ can be represented by a string of length polynomial in $\qr$ and
	$I$. Further, every node in $\dom_n$ can be represented by a string of
	length~$n$.
	
	By definition of $\cophtwo$, it suffices to give an algorithm with polynomial
	time bound in $\langle \qr,I,n \rangle$ such that there is some
	valuation~$V$ such that for every valuation $V'$, every node~$\node$ and every
	string~$\addInput$ the algorithm accepts input
	\begin{math}
		(
			\langle \qr,I,n\rangle,
			V,
			\langle V',\node,\addInput\rangle
		)
	\end{math}
	if and only if $(\qr,I,n) \notin \PCI(\classpnrel)$. In the following, we
	describe the algorithm and mention the time complexity of its main steps.
	Correctness of the algorithm results from the characterization above.
	
	First, the algorithm tests whether $V$ and $V'$ satisfy~$\qr$ on $I$. If $V$
	does not satisfy on $I$, the algorithm rejects. If $V'$ does not satisfy on
	$I$, the algorithm accepts. 
	Second, the algorithm tests whether $V$ and $V'$ derive the same fact,
	$V(\head{\qr}) = V'(\head{\qr})$. If they do not, the algorithm accepts.
	Otherwise it continues. All these tests can
	obviously be done in time polynomial in $\qr$ and $I$.

	We assume $\addInput$ to be a string of the form $\addInputk$, where each
	$\addInputi$ is a string of length polynomial in $n$ and $k$ is the number of
	atoms in $\qr$.
	
	Lastly, as described in Remark~\ref{rem:check-assign}, the algorithm tests
	whether there is a fact in $V'(\body{\qr})$ that is \emph{not} assigned to node~$\node$ and	passes the result (if a fact is not assigned to $\node$, the algorithm 	accepts, otherwise it rejects). Since the size of $V'(\body{\qr})$ is
	polynomially bounded by $\qr$ and $I$, this step also has a time bound
	polynomial in $\langle \qr,I,n \rangle$ (recall that $n$ is the representation
	size of nodes in the network underlying $\distp_n$).
\end{proof}

\begin{corollary}
	\label{cor:pci-fin-in-phtwo}
	Problem $\PCI(\classpfin)$ is in $\phtwo$.
\end{corollary}

\begin{proof}
	Obviously, problem $\PCI(\classpfin)$ can be polynomially reduced to
	$\PCI(\classpnrel)$ by choosing $n$ to be the maximum representation length of
	any node in $\distp$.
	Even though $n$ has to be encoded unary, explicit enumeration of all fact-node assignments in any
	distribution policy from $\classpfin$ guarantees that the length of $n$ is
	polynomial in $\distp$, which is an input parameter for problem
	$\PCI(\classpfin)$.
\end{proof}

\medskip
\noindent

To decide parallel-correctness of a query under a distribution policy for all
subinstances of a given instance, we use a minor variation of
Lemma~\ref{lem:pc}, where parallel-correctness is not required for all possible
instances but rather for all subinstances of a given instance.

\begin{lemma}
	\label{lem:pc-subinstances}
	A query~$\qr$ is parallel-correct on all subinstances~$I'$ of an instance~$I$
	under distribution policy~$\distp$ if and only if for every minimal valuation
	$V$ that satisfies $\qr$ on $I$ it holds
	\begin{math}
		\bigcap_{\fc \in V(\body{\qr})} \distp(\fc) \neq \emptyset.
	\end{math}
\end{lemma}

\medskip
\begin{proposition}
	\label{prop:pc-nrel-in-phtwo}
	Problem $\PC(\classpnrel)$ is in $\phtwo$.
\end{proposition}

\begin{proof}
	We show that the complement problem $\overline{\PC(\classpnrel)}$ is
	in $\cophtwo=\coClass{\phtwo}$. Now, instance~$I$ is an arbitrary instance
	whose data values can be represented by a string of length~$n$.
	
	By Lemma~\ref{lem:pc-subinstances}, a query~$\qr$ is not parallel-correct on
	each subinstance~$I'$ of $I$ under distribution policy~$\distp_n$ if there is a
	minimal valuation~$V$ that satisfies $\qr$ on $I$ such that no node in $\dom_n$
	is assigned all facts from $V(\body{\qr})$.
	
	Obviously, valuations can be represented by a string of length polynomial in
	$\qr$ and $n$. Further, every node in $\dom_n$ can be represented by a
	string of length~$n$.

	By definition of $\cophtwo$, it suffices to give an algorithm with polynomial
	time bound in $\langle \qr,n \rangle$ such that there is some
	valuation~$V$ such that for every valuation $V'$, every node~$\node$, and every string~$\addInput$ the algorithm accepts
	input
	\begin{math}
		(
			\langle \qr,n \rangle,
			V,
			\langle V',\node,\addInput \rangle
		)
	\end{math}
	if and only if $(\qr,n) \notin \PC(\classpnrel)$. In the following we
	describe the algorithm and mention the time complexity of its main steps.
	Correctness of the algorithm results from the characterization in Lemma~\ref{lem:pc-subinstances}.

	First, the algorithm tests whether valuation $V'$ contradicts minimality of
	valuation $V$. For this, it checks whether $V'$ derives the same fact,
	$V'(\head{\qr}) = V(\head{\qr})$, and also whether $V'(\body{\qr}) \subsetneq
	V(\body{\qr})$. If both tests succeed, the algorithm rejects. Otherwise it
	continues.
	
	Next, the algorithm tests whether a fact in $V(\body{\qr}) =
	\{\fc_1,\dots,\fc_\ell\}$ is not \emph{not} assigned to node~$\node$, as
	described in Remark~\ref{rem:check-assign}, and passes the result (if a fact
	is not assigned to $\node$, the algorithm accepts, otherwise it rejects). Since
	the size of $V(\body{\qr})$ is polynomially bounded by $\qr$ and $n$, this step
	also has a time bound polynomial in $\langle \qr,n \rangle$.
\end{proof}

\begin{corollary}
	\label{cor:pc-fin-in-phtwo}
	Problem $\PC(\classpfin)$ is in $\phtwo$.
\end{corollary}

\begin{proof}
	The algorithm for $\PC(\classpnrel)$ in the proof of
	Proposition~\ref{prop:pc-nrel-in-phtwo} can be adapted to a decision algorithm
	for $\PC(\classpfin)$. Before the other tests, it is checked whether valuations
	$V,V'$ satisfy $\qr$ on $\facts{\distp}$. If $V$ does not satisfy $\qr$ on
	$\facts{\distp}$, then the algorithm rejects. If $V'$ does not satisfy $\qr$
	on $\facts{\distp}$, then the algorithm accepts. In any other case, the
	algorithm continues.
	
	Obviously, both tests can be accomplished in time polynomial in $\langle \qr,
	\distp \rangle$.
\end{proof}

\subsubsection{Lower bounds}

We first prove the lower bound for $\PCI(\classpfin)$ in
Proposition~\ref{prop:pci-fin-phtwo-hard} by a reduction from $\twoTQBF$. This
reduction can also be adapted to a polynomial reduction to $\PC(\classpfin)$,
Proposition~\ref{cor:pc-fin-phtwo-hard}.

\begin{proposition}
	\label{prop:pci-fin-phtwo-hard}
	 $\PCI(\classpfin)$ is $\phtwo$-hard, even for distribution policies with
     only two nodes.
\end{proposition}

\begin{proof}
	We give a polynomial reduction from the $\phtwo$-complete problem
	$\twoTQBF$.

 	Let $\phi$ be an input for $\twoTQBF$, i.e., a formula of the form
	\begin{math}
		\forall \vx \exists \vy\, \psi(\vx,\vy).
	\end{math}
	We assume $\psi$ to be a propositional formula in 3-CNF with variables
	$\vx=(\seqx)$ and $\vy=(\seqy)$. Let $C_1,\dots,C_\lc$ denote their
	(disjunctive) clauses, where, for each $j$, $C_\itc=(\lit{\itc}{1} \lor \lit{\itc}{2} \lor
	  			\lit{\itc}{3})$. 

	We describe next how the corresponding input instance for
        $\PCI(\classpfin)$, consisting of a query~$\dquery$, a database
        instance~$\dinstance$, and a distribution policy~$\dpolicy$ is defined. 

	The query~$\dquery$ is formulated over variables $\vart$, $\varf$, and
	$x_\itx,\nx_\itx, y_\ity,\ny_\ity$, for $\itx\in\{\indx\}$ and
	$\ity\in\{\indy\}$. Intuitively, these variables are intended to represent the
	Boolean values true and false and the (negated) values of the variables
	$x_\itx,y_\ity$ in $\psi$, respectively.  We overload the notation $\ell_{j,i}$ as follows: if $\ell_{j,i}$
    is a negated literal $\neg x$ in $C_j$, then
    $\ell_{j,i}$ also denotes the variable $\nx$.
	
	Let $\threePosBool\mydef\threeBool \setminus \{(0,0,0)\}$
	be the set of non-zero Boolean triples and $\threePosVar \mydef \Var
	\setminus \{(\varf,\varf,\varf)\}$ the set of triples over
	$\{w_0,w_1\}$ that contain at least one $w_1$.

	We define $\dquery$ to be the query with $\head{\dquery}=H(\seqx)$ and
	\begin{math}
	  	\body{\dquery} = \consistency{3-CNF} \cup \structure{3-CNF}{\psi},
	\end{math}
	where
	\begin{displaymath}
		\consistency{3-CNF}
		\mydef
		\big\{\True(\vart), \False(\varf),\Neg(\vart,\varf), \Neg(\varf,\vart)\big\}
	  	\cup
	  	\big\{
	  		\Clause_\itc(\vw) \mid \itc \in \{\indc\}, \vw \in \posVar
	  	\big\}
	\end{displaymath}
	is a set of consistency atoms, representing valid combinations of values for
	$\Neg$-facts and satisfying combinations of values for $\Clause_\itc$-facts,
	and
	\begin{displaymath}
		\begin{array}{lll}
	  		\structure{3-CNF}{\psi} & \mydef & \big\{\Neg(x,\nx) \mid x \in
	  		\{\seqx,\seqy\}\big\}, \\
	  			& \cup & \big\{\Clause_\itc(\seqlit{\itc}) \mid
	  			\text{for each clause $C_\itc=(\lit{\itc}{1} \lor \lit{\itc}{2} \lor
	  			\lit{\itc}{3})$}\big\}
	  			\\
	  	\end{array}
	\end{displaymath}
	is a set of atoms representing the logical structure of $\psi$: it relates
	variable $x_\itx$ to $\nx_\itx$ and also variable $y_\ity$ to $\ny_\ity$ for
	each $\itx \in \{\indx\}$ and $\ity \in \{\indy\}$, respectively. 
	Additionally,
	it relates all variables that represent literals occurring in the same clause
	to each other.,
The database instance  $\dinstance$ is defined as
	\begin{displaymath}
	  	\big\{\True(1), \False(0), \Neg(1,0), \Neg(0,1)\big\} \cup
	  	\big\{\Clause_j(\vb) \mid j \in \{1,\dots,\lc\}, \vb \in
	  	\threeBool
	  	\big\}.
	\end{displaymath}
	 We partition it into $\negFacts \mydef \{\Clause_j(0,0,0)
         \mid j \in \{1,\dots,\lc\}\}$ and $\necFacts \mydef
         \dinstance \setminus \negFacts$. 

We	define $\dpolicy$ to be the finite distribution policy for $\dinstance$ over
	  	a network $\nw=\{\nodeA,\nodeB\}$ as
	\begin{displaymath}
	  	\dpolicy(\fc) = \left\{
	  		\begin{array}{ll}
	  			\{\nodeA\} & \text{if } \fc \in \necFacts, \\
	  			\{\nodeB\} & \text{if } \fc \in \negFacts. \\
	  		\end{array}
	  	\right.
	\end{displaymath}

	\paragraph*{Correctness}
	Obviously, query $\dquery$, instance $\dinstance$ and distribution policy
	$\dpolicy$ can be computed in polynomial time from
	$\phi$.
	We will now prove that this mapping is indeed a reduction,
        that is, that
	$(\dquery,\dinstance,\dpolicy) \in \PCI(\classpfin)$ if and only if
	$\phi\in\twoTQBF$.
	
	\subparagraph*{(if)}
	Let $\phi\in\twoTQBF$. We need to show  $(\dquery,\dinstance,\dpolicy) \in
	\PCI(\classpfin)$. To this end, let $\fc=H(a_1,\dots,a_\lx)$ be an arbitrary
	fact in $\dquery(\dinstance)$. We show that $\fc\in
        \dquery(\necFacts)$ and thus that $\fc$ is derived at node $\nodeA$.

	Since $\adom{I}=\{0,1\}$, it holds $a_1,\dots,a_\lx\in\{0,1\}$. Therefore, the
	mapping $\beta_\vx$ defined by $\beta_\vx(x_\itx)=a_\itx$ for each $\itx \in
	\{\indx\}$ is a well-defined \ta. By assumption, there is a \ta $\beta_\vy$ for
	$\seqy$ such that $\beta \models \psi$, where $\beta\mydef\beta_\vx \cup \beta_\vy$.
    
    We define valuation $V$ for $\dquery$ by $\vart
    \mapsto 1$, $\varf \mapsto 0$ and $u \mapsto \beta(u)$, $\negu \mapsto
    \overline{\beta(u)}$ for $u \in \{\seqx,\seqy\}$, respectively.
    Since $\beta\models\psi$, in particular, $\beta\models C_\itc$ for
    every $\itc \in
    \{\indc\}$. Therefore, for every  $\itc \in
    \{\indc\}$ there is some $\vb\in \threePosBool$ such that $V(\Clause_\itc(\seqlit{\itc})) =
    \Clause_\itc(\vb)$. Hence, all required $\Clause_\itc$-facts and $\True$-,
    $\False$-, $\Neg$-facts in $V(\body{\dquery})$ are contained in
    $\necFacts$ and therefore $H(a_1,\dots,a_\lx)=V(\head{\dquery})$
    can be derived on node~$\nodeA$.
	  
	\subparagraph*{(only-if)}
	Now, let $\phi\notin\twoTQBF$. We show that $(\dquery,\dinstance,\dpolicy)
	\notin \PCI(\classpfin)$. By the assumption on $\phi$ there is a valuation
	$\beta_\vx$ for $\seqx$ such that  for every valuation
	$\beta_\vy$ for $\seqy$ it holds $(\beta_\vx \cup \beta_\vy) \not\models \psi$.
	
	Let $\misFact=H(\beta_\vx(x_1),\dots,\beta_\vx(x_\lx))$ and $\beta$ be the
	extension of $\beta_\vx$ on the variables $\seqy$ defined by $\beta(y_\ity)=0$
	for every $\ity \in \{\indy\}$. This truth assignment induces a valuation $V$
	that satisfies $\dquery$ on $\dinstance$ because $\Clause_\itc(\vb)$-facts for all
	tuples $\vb \in \threeBool$ are in $\dinstance$. Thus, $\misFact \in \dquery(\dinstance)$.
	
	Clearly, $\misFact$ can not be derived at node
        $\nodeB$ as, e.g., $\negFacts$ does not contain any
        $\Neg$-facts. 
	However, $\misFact$ can neither be derived at node $\nodeA$:
        Towards a contradiction, let us assume that some
        valuation $V$ yields $\misFact$ at node $\nodeA$. In
        particular, $V$ needs to map all $C_j$-atoms of
        $\structure{3-CNF}{\psi}$ to facts from $\necFacts$. But then
        $V$ induces a truth assignment $\beta_\vy$ on variables $\seqy$ with
	$(\beta_\vx\cup\beta_\vy)\models\psi$, contradicting the choice of $\beta_\vx$.

\end{proof}

\begin{proposition}
	\label{cor:pc-fin-phtwo-hard}
	Problem $\PC(\classpfin)$ is $\phtwo$-hard, even for distribution policies with
     only two nodes.
\end{proposition}

\begin{proof}[Proof]
Again, we give a polynomial reduction from the $\phtwo$-complete problem
	$\twoTQBF$. 	Let $\phi$ be an input formula for $\twoTQBF$. The (intended) reduction maps
        $\phi$ to the query $\dquery$ and the distribution policy
        \dpolicy, as defined in the proof of
        Proposition~\ref{prop:pci-fin-phtwo-hard}. 

This mapping is clearly computable in polynomial time and we show next that
it is a reduction, that is, 	$(\dquery,\dpolicy) \in \PC(\classpfin)$ if and only if
	$\phi\in\twoTQBF$.

	\smallskip
	\noindent
	(if)
Let $\phi\in\twoTQBF$. We have to show that $\dquery$ is
parallel-correct on all input instances
        under $\dpolicy$.

	Let $I$ be such an instance. As $\True(1)$ and $\False(0)$ are
        the only facts under \dpolicy for relation symbols $\True$ and $\False$,
        respectively, every satisfying valuation for
        $\dquery$ needs to map $w_0\mapsto 0$ and $w_1\mapsto 1$.

If $\necFacts \not\subseteq I$, then
	$\dquery(I)=\emptyset$ because every satisfying valuation for $\dquery$ requires all facts from $\necFacts$ to satisfy the
	set $\consistency{3-CNF}$. Thus, in this case, $\dquery$  is parallel-correct on $I$ under
	distribution policy~$\dpolicy$.
	
If $\necFacts \subseteq I$, we have $\necFacts \subseteq I \subseteq
\dinstance$. In the proof of Proposition~\ref{prop:pci-fin-phtwo-hard}
it was shown that
$\dquery(\necFacts)=\dquery(\dinstance)$. Monotonicity yields
$\dquery(\necFacts)=\dquery(I)$ and therefore all facts in
$\dquery(I)$ are produced at node \nodeA.
	
	\smallskip
	\noindent
	(only-if)
	Let $\phi\notin\twoTQBF$. It was shown in  the proof of
        Proposition~\ref{prop:pci-fin-phtwo-hard} that in this case
        $\dquery$ is not parallel-correct on instance $\dinstance$
        under $\dpolicy$ and, in particular, it is not parallel-correct 
        under $\dpolicy$.
\end{proof}


%% file: appendix-transferability.tex
\section{Proofs for Section \ref{sec:trans}: Transferability}


\subsection{Proof of Lemma~\ref{lem:trans-sem-char}.}
\label{appendix:sem:trans}

The two implications of Lemma~\ref{lem:trans-sem-char} are shown in
Propositions~\ref{prop:trans-char-d1} and \ref{prop:trans-char-d2} below.

\begin{proposition} Let $\qr$ and $\qr'$ be CQs. If 
    condition (C2) holds, 
    then, parallel-correctness transfers from
    $\qr$ to $\qr'$.\label{prop:trans-char-d1} \end{proposition}

\begin{proof} 
Let $\distp$ be a distribution policy under which $\qr$ is parallel-correct
and let $I$ be an instance. Then we show that $\qr'$ is parallel-correct as well on $I$ under \distp. By monotonicity of CQs, $\bigcup_{x\in\nw}\qr'(\dist I (x)) \subseteq \qr'(I)$. Thus it suffices to show that for every fact $\fc \in \qr'(I)$, there is some valuation for $\qr'$
that allows to derive \fc on one of the chunks of $I$ under \distp.
For $\fc \in \qr'(I)$, there is a \inducedminimal valuation $V'$ for $\qr'$ which satisfies on $I$ for $\qr'$ and derives \fc. That is, $V'(\body{\qr'}) \subseteq I$ and
$V'(\head{\qr'}) = \fc$.  Next, we show that the facts required by $V'$ for $\qr'$ meet at
some node under \distp, which implies that the chunks of $I$ under \distp indeed allow
deriving \fc. 

For this, we rely on the assumption that there is a \inducedminimal valuation
$V$ for $\qr$, where $V'(\body{\qr'}) \subseteq V(\body{\qr})$. 
Let $J = V(\body{\qr})$. Then, by parallel-correctness of $\qr$ under \distp, there is a
valuation $W$ and node $\node \in \nw$, such that $W(\body{\qr}) \subseteq \dist
J (\node)$ and $W(\head{\qr}) = V(\head{\qr})$.  Because $V$ is \inducedminimal
and $\dist J (\node) \subseteq V(\body{\qr})$, it must be that $V(\body{\qr}) =
W(\body{\qr})$.
So, \distp maps all the facts in $J$ onto node $\node$, implying that all the
facts in $V'(\body{\qr'})$ are mapped onto node $\node$ under \distp (because
$V'(\body{\qr'}) \subseteq V(\body{\qr}) = J$).

Hence, $\qr'$ is indeed parallel-correct under the distribution policies for which $\qr$ is parallel-correct.
\end{proof}

\begin{proposition} Let $\qr$ and $\qr'$ be CQs. If parallel-correctness transfers from
    $\qr$ to $\qr'$, then, 
    condition (C2) holds.
    \label{prop:trans-char-d2} \end{proposition} \begin{proof} The proof
    is by contraposition.  So, we assume that there is a \inducedminimal valuation $V'$
    for $\qr'$ for which there is no valuation $V$ for $\qr$, where
    $V'(\body{\qr'})\subseteq V(\body{\qr})$. Let $m = |V'(\body{\qr'})|$.

We distinguish two cases, depending on whether $V'$ requires only one fact or at
least two facts.  For both cases we construct a network $\nw$ and distribution policy \distp over $\nw$ for which $\qr$ is parallel-correct but $\qr'$ is not, implying that
parallel-correctness does \emph{not} transfer from $\qr$ to $\qr'$. 

\begin{enumerate}

    \item \underline{Case ($m = 1$):} Let $V'(\body{\qr'}) = \{\fc\}$.  Let $\nw$ be a
        single-node network, i.e., $\nw \mydef \{\node\}$.  For $\distp$ we consider the
        distribution policy thats skips $\fc$, that is, maps $\distp(\fc)$ to the empty set, and maps every other fact in
        $\facts \sch$ onto node $\node$.
        By assumption on $V'$, none of the \inducedminimal valuations for $\qr$ requires $\fc$. So it
immediately follows by Lemma~\ref{lem:pc} that $\qr$ is parallel-correct under \distp.
However, because $V'$ is minimal for $\qr'$,  $\qr'$ needs $\fc$ to derive
$V(\head{\qr'})$ when only $\fc$ is given as input instance. Thus $\qr'$ is not
parallel-correct under $\distp$ which leads to the desired contradiction.

    \item \underline{Case ($m \ge 2$):}
Let $I\mydef V'(\body{\qr'}) = \{\fc_1, \ldots, \fc_m\}$, $\nw \mydef \{\node_1,\ldots,
\node_m\}$, and let \distp be the mapping defined as follows:
\begin{itemize}
\item $\distp(\fcB) = \nw$, for every $\fcB\in\facts{\sch}\setminus I$; and
\item $\distp(\fc_i) = \nw \setminus \{\kappa_i\}$, for every $i$.
\end{itemize}
Intuitively, on every instance $J$, either the facts in $J$ meet on some node under $\distp$, or $I \subseteq J$. By assumption, none of the \inducedminimal valuations for \qr requires all the facts in $I$, implying that \qr is parallel-correct under \distp.
Nevertheless, on instance $I$ under $\distp$, none of the nodes receives all the facts in $I$, and there is no valuation that can derive $V'(\head{\qr'})$ for a strict subset of the facts in $I$ (by \inducedminimal{}ity of $V'$). So, $\qr'$ is not parallel-correct under \distp which leads to the desired contradiction.
\end{enumerate}
\end{proof}

\begin{remark}
The reader may wonder what the effect is on Lemma~\ref{lem:trans-sem-char} when
distribution policies are not allowed to skip facts.
In fact, not much changes, except that we now need to deal explicitly with the trivial
case where a valuation for $\qr'$ requires only one fact.
More specifically, in this case, $\qr'$ is parallel-correct under all the distribution policies where $\qr$ is parallel-correct for if and only if 
\begin{itemize}
    \item[(C2')] for every minimal valuation $V'$ for $\qr'$, either $V'$ requires only
        one fact, or there is a minimal valuation $V$ for $\qr$ such that $V'(\body{\qr'})
        \subseteq V(\body{\qr})$. 
\end{itemize}

The proof is similar to the proof for Proposition~\ref{prop:trans-char-d1} and
Proposition~\ref{prop:trans-char-d2}, except that when condition (C2') fails, a
counterexample must require at least two facts, and thus case~1 in the proof of
Proposition~\ref{prop:trans-char-d2} drops. Notice that the
distribution policy given in case~
2
does not skip any facts, and is thus well-defined in the adapted setting as well.
\end{remark}




\newpredicate{\Con}{And}
\newpredicate{\Dis}{Or}
\newpredicate{\Res}{Res}
\newpredicate{\Val}{Val}
\newpredicate{\XVal}{XVal}
\newpredicate{\YVal}{YVal}
\newcommand{\XValVar}[1]{\EM{\XVal_{#1}}}
\newcommand{\YValVar}[1]{\EM{\YVal_{#1}}}

\newmathabb{\xdomain}{\mathit{Fix}}
\newmathabb{\gates}{\mathit{Gates}}
\newmathabb{\circuit}{\mathit{Circuit}}

\newmathabb{\lasti}{u}

\newmathabb{\enc}{\textit{enc}}
\newmathabb{\dec}{\textit{dec}}
\newmathabb{\cons}{a}
\newmathabb{\consTrue}{c_1}
\newmathabb{\consFalse}{c_0}

\newcommand{\querytransfrom}{\EM{\qr_\phi}}
\newcommand{\querytransto}{\EM{\qr'_\phi}}
\newcommand{\valtransfrom}{\EM{V}}
\newcommand{\valtransto}{\EM{V'}}
\newcommand{\valtransfromce}{\EM{V^\ast}}
\renewcommand{\uniV}{\EM{V'}}
\renewcommand{\exisV}{\EM{V}}
\newmathabb{\imuniV}{W'}
\newmathabb{\imexisV}{W}
\newmathabb{\imuniS}{s_W}
\newmathabb{\imexisS}{s_{W'}}

\newcommand{\bodyfrom}{\EM{\body{\querytransfrom}}}
\newcommand{\bodyto}{\EM{\body{\querytransto}}}
\newcommand{\headfrom}{\EM{\head{\querytransfrom}}}
\newcommand{\headto}{\EM{\head{\querytransto}}}
\newcommand{\evalheadfrom}{\EM{\valtransfrom(\headfrom)}}
\newcommand{\evalheadfromce}{\EM{\valtransfromce(\headfrom)}}
\newcommand{\evalheadto}{\EM{\valtransto(\headto)}}
\newcommand{\evalbodyfrom}{\EM{\valtransfrom(\bodyfrom)}}
\newcommand{\evalbodyfromce}{\EM{\valtransfromce(\bodyfrom)}}
\newcommand{\evalbodyto}{\EM{\valtransto(\bodyto)}}

\newcommand{\conttofrom}{\EM{\evalbodyto \subseteq \evalbodyfrom}}
\newcommand{\conttofromce}{\EM{\evalbodyto \subseteq \evalbodyfromce}}


\subsection{Proof of Theorem~\ref{theo:trans-completeness}}
\label{appendix:trans-compl}

In principle, Lemma~\ref{lem:trans-sem-char}, on which the following proofs are
based, talks about an infinite number of valuations over the infinite domain
$\dom$. However, since our queries are generic, the only observable property of
the constants used by some valuation is equality/inequality.
It therefore suffices to check valuations over an arbitrary finite domain with
at least as much constants as valuations for both queries can use. This is
stated more explicitly in the following claim.

\begin{claim}
	\label{lem:im-val-fin-dom}
	Let $\qr$ and $\qr'$ be CQs with variables $\seqx$ and $\seqy$, respectively.
	Moreover, for $k=\lx+\ly$ let $\domk=\{1,\dots,k\}$ be a subset of the
	(countably) infinite set $\dom$.

	The following two conditions are equivalent.
	\begin{enumerate}
	  \item For every \imv $V'$ for $\qr'$ over $\dom$ there is \animv
	  $V$ for $\qr$ over $\dom$ such that $V'(\body{\qr'}) \subseteq
	  V(\body{\qr})$.
	  \item For every \imv $V'$ for $\qr'$ over $\domk$ there is \animv
	  $V$ for $\qr$ over $\domk$ such that $V'(\body{\qr'}) \subseteq
	  V(\body{\qr})$.
	\end{enumerate}
\end{claim}

\bigskip
\noindent
Completeness of parallel-correctness transfer for general distribution
policies follows from the upper bound stated in
Proposition~\ref{prop:trans-upper-bound} and the lower bound stated in
Proposition~\ref{prop:trans-lower-bound}.

\begin{proposition}
	\label{prop:trans-upper-bound}
	$\PCTRANS \in \phthree$.
\end{proposition}

\begin{proof}
	By Lemma~\ref{lem:trans-sem-char}, deciding parallel-correctness transfer is
	equivalent to verifying that for each \imv $V'$ for $\qr'$
	there is \animv $V$ for $\qr$ such that $V'(\body{\qr'})
	\subseteq V(\body{\qr})$. This, in turn, is equivalent to checking for each
	valuation~$V'$ for $\qr'$ that it is not minimal, which can be witnessed by
	another valuation $W'$ that derives the same fact and requires strictly less
	facts, \emph{or} that there is \animv $V$ for $\qr$ such that $\conttofrom$.
	Non-minimality of valuation~$\exisV$ can be witnessed by a valuation~$\imexisV$.
	Due to Claim~\ref{lem:im-val-fin-dom}, all
	valuations can be restricted to $\domk=\{1,\dots,k\}$, where $k=\lx+\ly$ and
	$\qr,\qr'$ are queries over variables $\seqx$ and $\seqy$, respectively.

	To prove membership in class $\phthree$, it suffices to show that there is an
	algorithm with a time bound polynomial in $|\qr|+|\qr'|$ such that for every
	pair $(\qr,\qr')$ of queries it holds $(\qr,\qr') \in \PCTRANS$ if and only if
	for every $\qr'$-valuation $\uniV$ there is a $\qr$-valuation $\exisV$ and a
	$\qr'$-valuation~$\imuniV$ such that for every $\qr$-valuation $\imexisV$, the
	algorithm accepts
	\begin{math}
		\big(
			\langle\qr,\qr' \rangle,
			\uniV,
			\langle \exisV,\imuniV \rangle,
			\imexisV
		\big).
	\end{math}

	For input
	\begin{math}
		\big(
			\langle\qr,\qr' \rangle,
			\uniV,
			\langle \exisV,\imuniV \rangle,
			\imexisV
		\big)
	\end{math}
	the algorithm proceeds as follows. First, it is checked whether $\imuniV$
	contradicts the assumed \imity of $\uniV$, that is, whether
	$\imuniV(\head{\qr'})=\uniV(\head{\qr'})$ as well as
	$\imuniV(\body{\qr'})\subsetneq\uniV(\body{\qr'})$. If this test succeeds, the
	algorithm accepts because there is no requirement on a non-minimal
	$\qr'$-valuation.
	Second, it is checked in an analogous fashion whether $\imexisV$ contradicts
	the assumed \imity of $\exisV$. If this test succeeds, the algorithm rejects.
	
	Lastly, the algorithm continues with testing
	$\uniV(\body{\qr'}) \subseteq \exisV(\body{\qr})$ and accepts in case of satisfaction, and rejects otherwise.
	All containment tests can be done in polynomial time.
\end{proof}

\begin{proposition}
	\label{prop:trans-lower-bound}
	$\PCTRANS$ is $\phthree$-hard.
\end{proposition}

\begin{proof}
	We give a polynomial reduction from $\threeTQBF$ to $\PCTRANS$. The proof is
	based on the characterization of parallel-correctness transfer
        by condition (C2) as stated in
	Lemma~\ref{lem:trans-sem-char}.
	
	\paragraph*{Reduction function}
	Let thus $\phi= 
		\forall \vx
		\exists \vy
		\forall \vz\,
			\psi(\vx,\vy,\vz)$
be a formula
	with a quantifier-free  propositional formula $\psi$ in 3-DNF over
	variables $\vx=(\seqx)$, $\vy=(\seqy)$, and $\vz=(\seqz)$.
        
Let $k$ be the number of clauses of $\psi$ and, for each $\itc \in \{\indc\}$, let
	\begin{math}
		C_\itc=(\lit{\itc}{1} \land \lit{\itc}{2} \land \lit{\itc}{3})
	\end{math}
denote the $\itc$-th (conjunctive) clause of $\psi$.

	The reduction function maps $\phi$ to a pair
        $(\querytransfrom,\querytransto)$ of CQs that will be
        described next. It will be obvious that this mapping can be
        computed in polynomial time. Similarly as in
	the proof of Proposition~\ref{prop:pci-fin-phtwo-hard}, 
        \querytransfrom uses the 
	variables $\vart,\varf$, which are intended to represent truth
        and falseness, respectively, the variables of $\psi$ and variables $\negu$,
        for each variable $u$ of $\psi$, representing the literal
        $\neg u$.\footnote{Again, if $\ell$ is a negated literal $\neg
          u$, we write $\ell$ also for
          $\negu$.} Furthermore, it uses the following variables 
        \begin{itemize}
        \item $s_j$, for every $j\in\{1,\ldots,k\}$, intended to
          represent the truth value of $C_j$, and
        \item $r_j$, for every $j\in\{1,\ldots,k\}$, intended to
          represent the truth value of $C_1\lor\cdots\lor C_j$. 
        \end{itemize}

	We first describe the general construction, give an example explaining its
	intuition afterwards and finally prove correctness of the
        reduction.

	The queries \querytransto and \querytransfrom are defined as follows:
	\medskip
	\noindent
	\begin{displaymath}
		\begin{array}{lll}
                  \headto & \mydef & H(\seqx,\vart,\varf)\\
			\bodyto & \mydef &
			\big\{
				\YValVar{1}(\vart),\YValVar{1}(\varf), \dots,
				\YValVar{\ly}(\vart),\YValVar{\ly}(\varf)
			\big\} \cup \{\Res(\vart)\} \\
			& \cup &
			\xdomain; \\
		 	\\
                  \headfrom & \mydef &H(\seqx,\seqy,\vart,\varf)\\
			\bodyfrom & \mydef &
			\{
				\YValVar{1}(y_1),\YValVar{1}(\ny_1), \dots,
				\YValVar{\ly}(y_\ly),\YValVar{\ly}(\ny_\ly)
			\} \cup \{\Res(\varf), \Res(r_k)\} \\
			& \cup &
			\xdomain
			\cup
			\gates
			\cup
			\circuit, \\
		\end{array}
	\end{displaymath}
	where
	\begin{displaymath}
		\xdomain \mydef \big\{
			\XValVar{1}(x_1), \dots, \XValVar{\lx}(x_\lx),
			\True(\vart), \False(\varf)
		\big\}
	\end{displaymath}
is intended to	\enquote{fix} truth values for $\seqx,\vart,\varf$, the set
	\begin{displaymath}
		\begin{array}{lll}
			\gates & \mydef &
				\{\Neg(\varf,\vart), \Neg(\vart,\varf)\} \\
				& \cup & \{\Con(\vart,\vart,\vart,\vart), \Con(\varf,\vart,\vart,\varf),
				\Con(\vart,\varf,\vart,\varf), \Con(\varf,\varf,\vart,\varf),\\
				& & \hspace{5mm} \Con(\vart,\vart,\varf,\varf), \Con(\varf,\vart,\varf,\varf),
				\Con(\vart,\varf,\varf,\varf), \Con(\varf,\varf,\varf,\varf)\} \\
				& \cup & \{\Dis(\vart,\vart,\vart), \Dis(\varf,\vart,\vart),
				\Dis(\vart,\varf,\vart), \Dis(\varf,\varf,\varf) \}\\
		\end{array}
	\end{displaymath}
contains all atoms that are consistent with respect to the intended
meaning of negation, And- and Or-gates\footnote{The last position in
  a gate-atom represents the output bit of the gate, the others the
  input bits.} on $\vart,\varf$, and
	\begin{displaymath}
		\begin{array}{lll}
			\circuit & \mydef &
				\{\Neg(u,\negu) \mid \text{for each variable $u$ in $\psi$}\} \\
				& \cup & \{\Con(\seqlit{j},s_j) \mid \text{for each
				clause } C_j=(\lit{j}{1} \land \lit{j}{2} \land \lit{j}{3})\} \\
				& \cup & \{\Dis(s_1,s_1,r_1), \Dis(r_1,s_2,r_2),
				\dots,
				\Dis(r_{k-1},s_k,r_k)\}
				\\
		\end{array}
	\end{displaymath}
is intended to represent a Boolean circuit (with output bit $r_k$) that evaluates $\psi$.
		
	\begin{example}
		\label{ex:transfer-reduction}
		For formula
		\begin{math}
			\phi =
			\forall x_1
			\exists y_1
			\exists y_2
			\forall z_1 \;
				\Big(
					(x_1 \land y_1 \land z_1)
					\,\lor\,
					(\neg x_1 \land y_2 \land z_1)
				\Big)
		\end{math}
		we obtain the queries
		\begin{displaymath}
			\begin{array}{llcl}
				\querytransto: & H(x_1,\vart,\varf) & \gets &
					\YValVar{1}(\vart),\YValVar{1}(\varf), \YValVar{2}(\vart),\YValVar{2}(\varf), \Res(\vart), \\
					&&&
					\XValVar{1}(x_1),
					\True(\vart),\False(\varf).\\
			\end{array}
		\end{displaymath}
		and
		\begin{displaymath}
			\begin{array}{llcl}
				\querytransfrom: & H(x_1,x_2,y_1,\vart,\varf) & \gets &
					\YValVar{1}(y_1),\YValVar{1}(\ny_1),\YValVar{2}(y_2),\YValVar{2}(\ny_2),\Res(\varf),\Res(r_2), \\
					&&&
					\XValVar{1}(x_1),
					\True(\vart),\False(\varf),\\
					&&&
					\textit{\dots all atoms from $\gates$ \dots},\\
					&&&
					\Neg(x_1,\nx_1),\Neg(y_1,\ny_1),\Neg(y_2,\ny_2),\Neg(z_1,\nz_1), \\
					&&&
					\Con(x_1,y_1,z_1,s_1),\Con(\nx_1,y_2,z_1,s_2),
					\Dis(s_1,s_1,r_1),\Dis(r_1,s_2,r_2). \\
			\end{array}
		\end{displaymath}
		
		Note that $\phi \notin \threeTQBF$ because no \ta with
             $z_1 \mapsto 0$ is satisfying for $\psi$. In particular,
             for the \ta $\beta_\vx: x_1\mapsto 1$ there is no \ta
             $\beta_\vy$ such that for every  $\beta_\vz$ it holds
             $(\beta_\vx \cup \beta_\vy \cup \beta_\vz) \models
             \psi$. 
We illustrate why
		$(\querytransfrom,\querytransto) \notin \PCTRANS$.
		
	Let $\valtransto$ be the valuation for $\querytransto$ defined
        by
		$\valtransto(x_1)\mydef \beta_\vx(x_1)=1$,
		$\valtransto(\vart)\mydef 1$ and
                $\valtransto(\varf)\mydef 0$. This valuation is \im for
		$\querytransto$ because $\querytransto$ is full.
		We get $\evalbodyto=
			\big\{
				\YValVar{1}(1),\YValVar{1}(0),\Res(1),
				\XValVar{1}(1),\True(1),\False(0)
			\big\}.$

		\smallskip
		\noindent
		We now argue why there is no \im valuation $\valtransfrom$ for
                $\querytransfrom$ such that
		$\conttofrom$. If a valuation $\valtransfrom$ fulfills
                $\conttofrom$ it must map $w_0\mapsto 0$,
                $w_1\mapsto 1$, $x_1\mapsto 1$, $r_2\mapsto
                1$. Furthermore, it must map each of $(y_1,\ny_1)$ and
                $(y_2,\ny_2)$ to some pair in $\{(0,1),(1,0)\}$. Thus, $V$
                induces a \ta $\beta_\vy$ via $\beta_\vy(y_1)\mydef
                V(y_1)$ and $\beta_\vy(y_2)\mydef
                V(y_2)$. Let $\valtransfromce$ be the valuation that coincides
                with $V$ on all variables
                $w_0,w_i,x_1,\nx_1,y_1,\ny_1,y_2,\ny_2$ and maps
                $z_1\mapsto 0$ and maps all other variables to the
                ``correct'' values with respect to the semantics of
                the logicals gates in $\querytransfrom$. In
                particular, since  $(\beta_\vx \cup \beta_\vy \cup \beta_\vz) \not\models
             \psi$ (where $\beta_\vz(z_1)\mydef
                0$), we get $\valtransfromce (r_2)=0$. It is now easy
                to check
                that $\valtransfromce<_\qr \valtransfrom$, and 
                therefore that $V$ is not \im. \qed
 	\end{example}

	\medskip
	\noindent
To complete the proof, we need to show the mapping $\phi\mapsto
(\querytransfrom,\querytransto)$ is indeed a reduction. 

We start by some observations. We call a valuation $V$ of
\querytransfrom 
\emph{0-1-valued}, if its range is $\{0,1\}$ and it maps
$(w_0,w_1)$ to $(0,1)$ and  every pair $(u,\negu)$ of variables from $\psi$ to $(0,1)$ or $(1,0)$. A 0-1-valued
valuation is called
\emph{consistent}, if the values $V(s_j)$ and $V(r_j)$, for $j\in\{1,\ldots,k\}$ are consistent
with the values $V(u)$ for variables of $\psi$, in the obvious
sense. That is, $V(s_j)=1$ if and only clause $C_j$ evaluates to true
for the truth assignment $\beta_V$ obtained from $V$ and $V(r_j)=1$ iff
$C_1\lor\cdots\lor C_j$ evaluates to true. 

It is easy to see that a 0-1-valued valuation $V$ is consistent, if and
only if $V(\circuit)\subseteq V(\gates)$, because inconsistency
requires facts in $V(\circuit)$ that are not in $V(\gates)$ and
likewise the existence of such facts implies inconsistency.

\begin{claim}\label{claim:red}
  For every 0-1-valued valuation $V$ of \querytransfrom the following
  conditions are equivalent.
  \begin{enumerate}[(i)]
  \item $V$ is \im;
  \item $V$ is consistent.
  \end{enumerate}
\end{claim}
To show that (i) implies (ii), let $V$ be a 0-1-valued \im valuation. Let $\valtransfromce$ be the uniquely determined
0-1-valued consistent valuation that agrees with $V$ on all variables of $\psi$. As \valtransfromce is consistent,
$\valtransfromce(\circuit)\subseteq
\valtransfromce(\gates)$. Since $V$ and \valtransfromce  agree
on all variables of $\psi$ and $V$ is \im, we get
$V(\circuit)\subseteq \valtransfromce(\gates)=V(\gates)$ and thus $V$
is consistent. 

To show that (ii) implies (i), let $V$ be a 0-1-valued consistent
valuation. On the other hand,
consistency implies $V(\circuit)\subseteq V(\gates)$. Towards a
contradiction, let us assume, there is a valuation \valtransfromce
with $\valtransfromce<_\qr\valtransfrom$.  
As $V$ and \valtransfromce agree on $w_0$, $w_1$ and all variables
$x_i,y_i$, the only possible fact that could account for the
strictness in $\valtransfromce(\bodyfrom)\subsetneq
                \evalbodyfrom$ is $\Res(1)$, in case $\valtransfrom (r_k)=1$ and  $\valtransfromce(r_k)=0$. However, $\valtransfromce(\circuit)\subseteq
                V(\circuit)\subseteq
                V(\gates)=\valtransfromce(\gates)$, where the middle
                inclusion holds because $V$ is
                consistent. Thus,
$\valtransfromce$
                is also consistent, and therefore
                $\valtransfromce(r_k)=\valtransfrom(r_k)$, the desired contradiction.\\

Now, we are prepared to prove that $\phi$ is in $\threeTQBF$ if and only if
	parallel-correctness transfers from $\querytransfrom$ to
        $\querytransto$.
 
	\subparagraph*{(only-if)}
	Let $\phi= 
		\forall \vx
		\exists \vy
		\forall \vz\,
			\psi(\vx,\vy,\vz)$
be a formula
	with a quantifier-free  propositional formula $\psi$ in 3-DNF  such that $\phi \notin \threeTQBF$. We show that there is \animv
	$\valtransto$ for $\querytransto$ such that each
	valuation $\valtransfrom$ for $\querytransfrom$ which satisfies
	$\conttofrom$ is \emph{not} \im. From that we can conclude by Lemma
        \ref{lem:trans-sem-char} that 	parallel-correctness does not transfer from $\querytransfrom$ to
        $\querytransto$.

Let $\beta_\vx$ be a \ta for $\seqx$ in $\psi$
	such that for all \tas $\beta_\vy$ for $\seqy$ in $\psi$ there
	is a \ta $\beta_\vz$ for $\seqz$ such that $(\beta_\vx \cup \beta_\vy \cup \beta_\vz) \not\models
	\psi$.

	Let $\valtransto$ be the valuation defined by
	$\valtransto(\seqx,\vart,\varf) \mydef
	(\beta_\vx(x_1),\dots,\beta_\vx(x_\lx),1,0)$, which is \im for $\querytransto$
	 because $\querytransto$ is full.

	Let $\valtransfrom$ be any valuation for $\querytransfrom$ such that
	$\conttofrom$. In particular,  $\Res(1)\in \valtransfrom(\bodyfrom)$.
Then, valuations $\valtransfrom$ and $\valtransto$
	agree on variables $\seqx,\vart,\varf$ because each atom in
        \xdomain  is the only  atom of $\querytransfrom$
        with its particular relation symbol. Similarly, the
        $\YValVar{i}$-atoms in $\querytransto$ and $\querytransfrom$
        ensure that $\valtransfrom$ maps each pair $(y_i,\ny_i)$ to
        $(0,1)$ or $(1,0)$. Let $\beta_\vy$ be the \ta defined by
        $\beta_\vy(y_i)\mydef \valtransfrom(y_i)$, for every $i\in
        \{1,\ldots,n\}$. Since $\phi \notin \threeTQBF$, there is a
        \ta $\beta_\vz$ such that $(\beta_\vx \cup \beta_\vy \cup \beta_\vz) \not\models
	\psi$. Let \valtransfromce be the uniquely defined consistent
        0-1-valued valuation induced by $(\beta_\vx \cup \beta_\vy
        \cup \beta_\vz)$. Since  \valtransfromce is consistent,
        $\valtransfromce (\circuit)\subseteq \valtransfromce
        (\gates)$ and therefore $\valtransfromce (\bodyfrom)\subseteq
                \evalbodyfrom$. Furthermore, since $(\beta_\vx \cup \beta_\vy \cup \beta_\vz) \not\models
	\psi$, we get $\valtransfromce(r_k)=0$ and
        therefore $\Res(1)\not\in \valtransfromce(\bodyfrom)$ and, consequently, $\valtransfromce (\bodyfrom)\subsetneq
                \evalbodyfrom$, showing that $V$ is not minimal.

	\subparagraph*{(if)}
Let now $\phi= 
		\forall \vx
		\exists \vy
		\forall \vz\,
			\psi(\vx,\vy,\vz)$ be a formula  in $\threeTQBF$ and let $\valtransto$ be an arbitrary
	valuation for $\querytransto$. We will show that there exists  \animv $\valtransfrom$ for
	$\querytransfrom$ such that $\valtransto(\body{\querytransto}) \subseteq
	\valtransfrom(\body{\querytransfrom})$, thus showing that
        parallel-correctness transfers from $\querytransfrom$  to 
        $\querytransto$, again by Lemma
        \ref{lem:trans-sem-char}.
	
We assume in the following that all quantified variables appear
(possibly negated) in
$\psi$. 
Let $c_0\mydef \valtransto(w_0)$ and $c_1\mydef
\valtransto(w_1)$. Since, neither $\querytransto$ nor $\querytransfrom$
uses any constant symbols, minimality of $\valtransto$ is not
affected, if $\valtransto$ is composed with any bijection of the
domain. The same holds for every valuation $\valtransfrom$ and the
statement 
$\valtransto(\body{\querytransto}) \subseteq
	\valtransfrom(\body{\querytransfrom})$, as long as
        $\valtransto$ and $\valtransfrom$ are composed with the
        \emph{same} bijection. Therefore, we can assume without loss
        of generality that $\valtransto(w_0)=0$ and
        $\valtransto(w_1)\in\{0,1\}$. 

We distinguish three cases depending on whether
$\dom(\valtransto)\subseteq\{0,1\}$ and $\valtransto(w_1)=1$.

	\begin{case}{1}{$\dom(\valtransto)\subseteq\{0,1\}$ and
            $\valtransto(w_1)=1$}
Let 	$\beta_\vx$  be the partial \ta
	for the variables $\seqx$ in $\psi$ defined by $\beta_\vx(x_i) =
		\valtransto(x_i)$, for every $i \in
		\{1,\dots,\lx\}$. 		  
	Since, $\phi \in \threeTQBF$, there exists a partial \ta
		$\beta_\vy$ for the variables $\seqy$ in $\psi$ such that for each partial
		\ta $\beta_\vz$ for the variables $\seqz$ we have
                $(\beta_\vx \cup \beta_\vy \cup \beta_\vz) \models
                \psi$. For concreteness let $\beta_\vz(z_\itz)\mydef 0$, for $i
                \in \{1,\dots,\lz\}$ and $\beta\mydef \beta_\vx \cup \beta_\vy \cup \beta_\vz$.
	
Let $\valtransfrom$ be the uniquely defined 0-1-valued consistent
valuation induced by $\beta$. Since $\valtransfrom$ is consistent it
is also minimal by Claim \ref{claim:red}, and as
$\beta\models \psi$, $\valtransfrom
(r_k)=1$. Thanks to the latter, $\conttofrom$ follows easily and
Case 1 is complete.

	\end{case}

	\begin{case}{2}{$\dom(\valtransto)\subseteq\{0,1\}$ and $\valtransto(w_1)=0$}
	  Let $\valtransfrom$ be defined by
\[
	  	\valtransfrom(u) \mydef
                \begin{cases}
                 \valtransto(u) & \text{if $u\in\{w_0,w_1,x_1,\ldots,x_n\}$},\\
                  0 & \text{otherwise}.
                \end{cases}
\]

It is easy to see that $\valtransto(\body{\querytransto}) \subseteq
	  \valtransfrom(\body{\querytransfrom})$. Furthermore,
          $\valtransto$ is minimal as every fact from 
          $\valtransfrom(\body{\querytransfrom})$ either stems from an
          atom with (only) head variables or is in the unavoidable set
          $\valtransfrom(\gates)$. 
	\end{case}

	\begin{case}{3}{For some $\mathbf{\itx}$,
	$\boldsymbol{\valtransto(x_\itx) \notin \{c_0,c_1\}}$}
We recall that by our assumptions, $c_0=0=\valtransto(w_0)$ and $c_1=
\valtransto(w_1)\in\{0,1\}$. The following argument works for both subcases, $c_1=1$
and $c_1=0$.
We call a variable $x_\itx$ \emph{foul} if $\valtransto(x_\itx) \notin
\{c_0,c_1\}$. Likewise, we call a clause  \emph{foul} if it contains
(positively or negatively) some foul variable. 
	  Let $\Itx$  be the set of all indices $\itx$ for which
          $x_\itx$ is foul and
	  $\Itc$ be the set of all indices $\itc$ of foul clauses. 
	  Furthermore, let\footnote{In fact, any value
	  not in $\{c_0,c_1\}$ would do.} $\cons=\valtransto(x_{\itx})$ for the
          minimal index $\itx\in\Itx$.

	\smallskip
	\noindent
	  We define valuation $\valtransfrom$ by 
\[
	  	\valtransfrom(u) \mydef
                \begin{cases}
                  \valtransto(u) & \text{if $u\in\{w_0,w_1,x_1,\ldots,x_m\}$},\\
                  c_1 & \text{if $u\in\{y_1,\ldots,y_n,z_1,\ldots,z_p\}$},\\
                  c_0 & \text{if
                    $u\in\{\ny_1,\ldots,\ny_n,\nz_1,\ldots,\nz_p\}$},\\
                  \cons & \text{if $x=\nx_\itx$ and $x_\itx$ is foul},\\
                  c_0 & \text{if $x=\nx_\itx$ and
                    $\valtransto(x_\itx) = c_1$},\\
                  c_1 & \text{if $x=\nx_\itx$ and
                    $\valtransto(x_\itx) = c_0$}.\\
                \end{cases}
\]
For variables $s_j$, $\valtransfrom(s_j)\mydef c_1$, if $C_j$ is foul or for all
its literals $\ell$, it holds $\valtransfrom(\ell)=c_1$, otherwise
$\valtransfrom(s_j)\mydef c_0$. For variables $r_j$,
$\valtransfrom(r_j)\mydef c_1$, if $\valtransfrom(s_i)= c_1$, for some
$i\le j$ and $\valtransfrom(r_j)\mydef c_0$, otherwise.

It is clear that \conttofrom holds, but  we can not expect that
\valtransfrom is minimal. There might be some \Con-facts in $\valtransfrom(\circuit)$ resulting
from clauses that can be avoided by changing the valuation for some
variables $z_i$. However, we can show in the following that every
minimal valuation \valtransfromce contained in \valtransfrom fulfills
\conttofromce and thereby yields (C3).

To this end, let \valtransfromce be a \im valuation such that
$\valtransfromce \le_\qr \valtransfrom$.
We show first that \valtransfromce has to produce most facts
from  \evalbodyfrom. This is immediate for all facts from $\valtransfrom(\{
				\YValVar{1}(y_1),\YValVar{1}(\ny_1), \dots,
				\YValVar{\ly}(y_\ly),\YValVar{\ly}(\ny_\ly)\})$, 
                                $\valtransfrom(\xdomain)$, and $\valtransfrom(\gates)$.

Any facts of the form $\valtransfrom(\Neg(u,\negu))$ that do not occur
in $\valtransfrom(\gates)$ are of the form
$\Neg(\valtransto(x_\itx),\cons)$, for some foul variable $x_\itx$. As
$x_\itx$ occurs in the head, and there is at most one such fact per
foul variable, these facts can not be avoided in
$\valtransfromce(\bodyfrom)$. As all facts of the form
$\valtransfromce(\Neg(u,\negu))$ have to be in
$\valtransfromce(\bodyfrom)$ and all variables $x_i,y_i$ occur in
\headfrom, we can conclude that \valtransfromce has to agree with
\valtransfrom for all variables of the form $x_i,\nx_i,y_i,\ny_i$ and
on $w_0$ and $w_1$. 

Therefore, it is clear for all facts from $\valtransto(\bodyto)$
except $\Res(c_1)$ that they are captured by
$\valtransfromce(\bodyfrom)$. It therefore only remains to show
$\Res(c_1)\in \valtransfromce(\bodyfrom)$.

Let $x_g$ be the foul variable that was used to define $a\mydef \valtransto(x_g)$
 and let $C_j$ be some clause in which it occurs. Thus, by definition
 of \valtransfrom there is an
 \Con-fact in $\valtransfrom(\circuit)$ with value $a$ in one of its first three positions and with
 $c_1$ in its fourth position. Furthermore, all
 \Con-facts in $\valtransfrom(\circuit)$
 with $a$-values have $c_1$ in their fourth position. Therefore,
 $\valtransfromce(\circuit)$ needs to contain at least one \Con-fact
 with $a$ in one of its first three positions and with $c_1$ in its
 fourth position. That is, $\valtransfromce(s_i)=c_1$, for at least
 one $i$. As $\valtransfromce(\circuit)$ can only contain \Dis-facts
 from $\valtransfrom(\gates)$, it follows that
 $\valtransfromce(r_h)=c_1$, for all $h\ge i$ and, in particular, for
 $h=k$.  Therefore, $\Res(c_1)\in \valtransfromce(\bodyfrom)$ and \conttofromce.
	\end{case}

\end{proof}

\subsection{Proof of Lemma~\ref{lem:comp_st_minimal}}
\label{ssec:Lemma_comp_st_minimal}

We define the following problem:

\problemdefinition
	{$\SMQ$}
	{CQ $\qr$}
	{Is $\qr$ strongly minimal?}

It remains to prove the following lemma:
\begin{lemma}
	$\SMQ$ is $\coNP$-hard.
\end{lemma}

\begin{proof}
	The proof is by a reduction from the $\NP$-complete $\ThreeSat$ problem to the complement of $\SMQ$. To this end, let $\phi$ be an input for $\ThreeSat$. That is, a propositional formula in
	3-CNF with variables $\seqx$. Let $\seqc$ denote the (disjunctive) clauses in $\varphi$. We next construct the corresponding input instance $\dquery$ for $\SMQ$.
		
	The head of $\dquery$ is $H(\vart,\varf,x_1,\nx_1,\dots,x_\lx,\nx_\lx)$. Intuitively, we represent Boolean values by
	pairs of variables: true is represented by $(\vart,\varf)$ and false is
	represented by $(\varf,\vart)$. Similarly, each literal of
$\phi$ is represented by a pair of variables.  Formally, this
representation is by a function	$\rep$, which maps literals from $\phi$ to pairs of
variables from $\qr$ via
	$\rep(x_\itx)\mydef(x_\itx,\nx_\itx)$ and $\rep(\neg x_\itx)\mydef(\nx_\itx,x_\itx)$, for each $\itx \in
	\{\indx\}$. We extend $\rep$ to 3-clauses via
	$\rep\big(\disClause{\itc}\big)\mydef(y_1,y'_1, y_2,y'_2, y_3,y'_3)$, where
	$(y_i,y'_i)=\rep(\lit{j}{i})$ for every $i \in \{1,2,3\}$.
	Besides the already mentioned head variables, $\dquery$ has only two non-head variables $r_0,r_1$.
	
	Let $\sixVar$ be the set of all 6-tuples $(u_1,u'_1,u_2,u'_2,u_3,u'_3)$,
	where $(u_i,u'_i)$ equals either $(\vart,\varf)$ or $(\varf,\vart)$ for each
	$i \in \{1,2,3\}$, and let $\sixPosVar =
	\sixVar \setminus \{(\varf,\vart,\varf,\vart,\varf,\vart)\}$.
Intuitively, $\sixPosVar$ represents all truth assignments that
satisfy a three-way disjunction.

	We define
	$\body{\dquery}\mydef \domain \cup \consistency{3-CNF} \cup \structure{3-CNF}{\phi}$,
	where
    \begin{itemize}\itemsep=0pt
    \item 
   
	\begin{math}
		\domain \mydef \{\Dom(r_0,r_1), \Dom(r_1,r_0)\}
	\end{math}
	is a set of binary atoms restricting the possible mappings of variables
	$r_0,r_1$; 

    \item 
	\begin{math}
		\consistency{3-CNF} \mydef
		\big\{
			\Clause_\itc (\vart, \varf,\vu) \mid \itc \in
			\{\indc\},\; \vu \in \sixPosVar \big\}
	\end{math}
	is a set of 8-ary atoms encoding satisfying valuations of clauses; and,
    \item 	\begin{math}
		\structure{3-CNF}{\phi} \mydef
		\big\{
			\Clause_\itc
			\big(
				r_1, r_0, \vy_\itc
			\big)
			\mid
			\itc \in \{\indc\},
			\text{ where }
			\vy_\itc\mydef	\rep(C_\itc)
		\big\}
	\end{math}
	is a set of 8-ary atoms representing the actual clauses of $\phi$.
	\end{itemize}

	\paragraph*{Correctness}
	Obviously,  $\dquery$ can be computed in polynomial time. It
        therefore only
	remains to show that $\dquery \notin \SMQ$ if and only if $\phi \in \ThreeSat$.
	
	\subparagraph*{(if)}
	Let $\phi \in \ThreeSat$. We show that $\dquery \notin \SMQ$ by constructing
	two valuations, that derive the same head fact, of which one requires strictly
	less facts than the other. By assumption, there is a \ta $\beta$ for variables
	$\seqx$ in $\phi$ such that $\beta \models \phi$.
	
	Let $\goodVal$ and $\badVal$ both be defined on the head variables
        as follows: $\varf
	\mapsto 0, \vart \mapsto 1$, and $x \mapsto \beta(x), \nx \mapsto \overline{\beta(x)}$
	for each $x \in \{\seqx\}$. The non-head variables
	are mapped for $\goodVal$ by $r_0 \mapsto 0, r_1 \mapsto 1$, and for
	$\badVal$ by $r_0 \mapsto 1, r_1 \mapsto 0$.
	
	Obviously, $\goodVal(\head{\dquery})=\badVal(\head{\dquery})$. Since
	$\beta$ is a satisfying \ta for $\phi$, it particularly satisfies every clause
	$C_\itc$ from $\phi$. Therefore,
	for every $\itc \in \{\indc\}$ there is some $\vu_\itc \in \sixPosVar$ such that
	\begin{math}
		\goodVal(\Clause_\itc(r_1,r_0,\vy_\itc))
		=
		\goodVal(\Clause_\itc(\vart,\varf,\vu_\itc)),
	\end{math}
	which implies $\goodVal(\structure{3-CNF}{\phi}) \subseteq
	\goodVal(\consistency{3-CNF})$. Together with
	$\goodVal(\domain)=\badVal(\domain)$ and
	$\goodVal(\consistency{3-CNF})=\badVal(\consistency{3-CNF})$ this implies
	$\goodVal(\body{\dquery}) \subseteq \badVal(\body{\dquery})$.
	Because $\badVal(r_1,r_0)=(0,1) \neq
        (1,0)=\badVal(\vart,\varf)$, it holds
	$\badVal(\structure{3-CNF}{\phi})\not\subseteq\badVal(\consistency{3-CNF})$,
	which further implies $\goodVal(\body{\dquery}) \subsetneq
	\badVal(\body{\dquery})$. Thus, $\goodVal$ contradicts \imity of $\badVal$ and
	therefore we have $\dquery \notin \SMQ$.
	
	\subparagraph*{(only-if)}
	Now, let $\phi$ be unsatisfiable. We show that $\dquery \in \SMQ$.
	Actually, we prove that there are even no two different valuations 
    $\val$ and $\altVal$ for $\dquery$ for which $V'\le_\dquery V$.
    
    Towards a contradiction, assume
    that  $V$ and $V'$ are two different valuations 
    for $\dquery$ with $V'(\head{\dquery})=V(\head{\dquery})$ and $V'(\body{\qr})\subseteq V(\body{\qr})$.
 By design of the set
	\begin{math}
		\domain = \{\Dom(r_0,r_1), \Dom(r_1,r_0)\},
	\end{math}
	 one of the following holds: (1) $\altVal(r_0)=\val(r_0)$ and $\altVal(r_1)=\val(r_1)$; or (2)
	$\altVal(r_0)=\val(r_1)$ and $\altVal(r_1)=\val(r_0)$.
	From our assumption $\altVal\neq\val$ and since $r_0,r_1$ are the only non-head-variables, we can conclude    that (2) must hold and that in addition $\val(r_0)\neq\val(r_1)$.
	
	To obtain a contradiction, we next show that there is a
        $\Clause_\itc$-atom $A \in \structure{3-CNF}{\phi}$ for which
        $V'(A)\not\in V(\body{\dquery})$ implying that  
$V'(\body{\dquery})\not\subseteq V(\body{\dquery})$. In fact, it already suffices to
	show that $\altVal(A) \notin \val(\consistency{3-CNF})$, since $A$ cannot be
	mapped to any fact in $\val(\structure{3-CNF}{\phi})$ because of the
	different mapping of variables $r_0,r_1$.

		Let $\consFalse \mydef V(\varf)$ and $\consTrue \mydef
                V(\vart)$. We distinguish two cases.

	\begin{case}{1}{$\val(r_1,r_0) = (c_1,c_0)$}
		Let $A\mydef\Clause_\itc(r_1,r_0,\vy_\itc)$ be an
                arbitrary $\Clause_\itc$-atom. Then $\altVal(\Clause_\itc(r_1,r_0,\vy_\itc)) \notin
		\val(\consistency{3-CNF})$ because $\altVal(r_1,r_0)\neq\val(\vart,\varf)$.
	\end{case}
	
	\begin{case}{2}{$\val(r_1,r_0) =  (c_0,c_1)$}
		
		We can assume that $V'(x_\itx,\nx_\itx) \in \{(\consFalse,\consTrue),\,
		(\consTrue,\consFalse)\}$ for every $\itx \in
                \{\indx\}$. Otherwise, existence of $A$ follows
                immediately, as  $\val(\consistency{3-CNF})$ only uses
                values $\consFalse,\consTrue$.

		\noindent
		Therefore, valuation $V'$
		induces a \ta $\beta$ by $x \mapsto 0$, if
		$\val(x)=\consFalse$, and $x \mapsto 1$, if $\val(x)=\consTrue$, for every $x
		\in \{\seqx\}$. Since $\phi$ is not satisfiable, there
                is some $\itc$ such that
		$\beta\not\models C_\itc$. Thus, we
		have $\beta(\seqlit{\itc})=(0,0,0)$ and
			$\altVal(\vy_\itc) = 
			(\consFalse,\consTrue,\consFalse,\consTrue,\consFalse,\consTrue)
			\notin \sixPosVar$.
	
		Therefore, valuation $\altVal$ cannot map atom $A=\Clause_\itc(r_0,r_1,\vy_j)$
		from $\structure{3-CNF}{\phi}$ to any fact in
		$\val(\consistency{3-CNF})$.
	\end{case}
	
\end{proof}


%% file: hc-complex-appendix.tex

\section{Proofs for Section \ref{sec:families}: Families of Distribution Policies}
\subsection{Proof of Proposition~\ref{prop:cthree}}
\label{app:proof-families-complexity}

\newcommand{\threecolor}{three-colorability}
\newcommand{\colors}{\EM{\textit{colors}}}
\newcommand{\graph}{\EM{\textit{graph}}}
\newcommand{\fixed}{\EM{\textit{fixed}}}
\newcommand{\acyclicity}{\EM{\textit{acyclicity}}}
\newcommand{\rest}{\EM{\textit{rest}}}
    \newcommand{\hg}{\EM{\mathcal{H}}}

\newcommand{\id}{\EM{\ell}}

We provide two reductions from an \NP-complete problem to the problem of deciding condition (C3) on arbitrary queries $\qr'$ and $\qr$.
In the first reduction, given in
Proposition~\ref{prop:hc-com-acyclic2}, $\qr$ is acyclic while $\qr'$
is acyclic in the second reduction, given in
Proposition~\ref{prop:hc-com-acyclic1}. 
Clearly, Proposition~\ref{prop:cthree} follows from
Proposition~\ref{prop:hc-com-acyclic1} and Proposition~\ref{prop:hc-com-acyclic2}. That the proposition holds
even for full \qr follows from the fact that the head of \qr  is
irrelevant for the truth of condition
(C3). That is, e.g., in the proof of
Proposition~\ref{prop:hc-com-acyclic1} the head of \qr could be chosen full.

For a CQ $\qr$, we denote by $\hg_\qr$ the hypergraph in which every node corresponds to a variable in $\qr$, and there is a hyperedge between a set of nodes in $\hg_\qr$ if the corresponding variables of $\qr$ occur together in some atom in the body of $\qr$. For the definition of acyclicity, we use the well-known GYO reduction~\cite{ahv_book}. In particular, a query $\qr$ is acyclic if repeatedly removing nodes in
$\hg_\qr$ that are in only one hyperedge and all the hyperedges that are
contained in another hyperedge, results in the empty hyperedge.

For both cases, \NP-hardness relies on a reduction from the \NP-complete graph
3-colorability problem, which asks for a given undirected graph $G=(V,E)$, whether there exists a mapping $h$ that assigns to every node in $G$ a color from the set $C =\{r, g, b\}$, such that a different color is assigned to every pair of adjacent nodes. 
Hereafter we will refer to such a mapping $h$ for $G$ as a 3-color assignment or mapping for $G$. 

For ease of exposition, we assume $V \subseteq \uvar$, $C \subseteq \uvar$ (where $V$ and $C$
are disjoint sets).
Further, by $E_C$ we denote the set of all pairs $(c,d) \in C^2$ where $c \ne d$.
Intuitively, these pairs correspond to possible valid colorings of edges.

\begin{proposition} For Boolean CQs \qr and $\qr'$, where $\qr$ is acyclic, it is \NP-hard
    to decide whether
    condition (C3) holds.
    \label{prop:hc-com-acyclic2}
\end{proposition} 

\begin{proof} As mentioned above, we reduce from the \NP-complete graph  $3$-colorability problem. To this end, let $G=(V,E)$ be an arbitrary input graph.

    Based on $G$, we construct a query $\qr'$ and $\qr$ as follows:\footnote{The construction is inspired by the construction given in
        \cite{DBLP:conf/stoc/ChandraM77} to show \NP-completes for the folding problem.}
        \begin{align*} 
            Q' & : () \leftarrow \bigwedge_{(x,y) \in E}E(x,y), \bigwedge_{(c,d) \in E_C} E(c,d), \Fix(r,g,b). \\ 
            Q & : () \leftarrow \bigwedge_{(c,d)\in E_C}E(c,d),\Fix(r, g, b).
        \end{align*}
 
        Here, both $\qr$ and $\qr'$ are defined over the schema $\{E,\fix\}$
        where $E$ is binary and $\fix$ is ternary. Notice in particular
        that the second part of the body of $\qr'$ corresponds exactly
        to $\body{\qr}$.  Furthermore, notice that $\qr'$ and $\qr$ 
        can be computed in time polynomial in the size of $G$.

    \paragraph{Claim.} There is a $3$-color assignment for $G$  if and only
    if there is a simplification $\theta$ for $\qr'$ and a substitution $\rho$ for
    $\qr$ such that $\body{\theta(\qr')} \subseteq \body{\rho(\qr)}$.

    \paragraph{(if)}
    Suppose there is a simplification $\theta$ for $\qr'$ and a substitution
    $\rho$ for $\qr$, such that $\body{\theta(\qr')} \subseteq \body{\rho(\qr)}$. By definition of simplification, 
    $\body{\theta(\qr')} \subseteq \body{\qr'}$.
    In particular, $\theta(\Fix(r, g, b)) \subseteq \body{\qr'}$, so it must be that
    $\theta(\Fix(r, g, b)) = \Fix(r, g, b)$. 
    So, $\Fix(r, g, b)$ must be in $\body{\rho(\qr)}$, implying that $\body{\rho(\qr)} = \body{\qr}$ as $\qr$ contains only colors. 
    This means that $\theta$ maps every node in $V$ onto a color in $C$. 
    Specifically, by construction of $\qr'$ and $\qr$, every two adjacent nodes
    in $G$ are mapped by $\theta$ onto distinct colors in $C$.  Now, let $h$ be the mapping defined as $h(x) = \theta(x)$ for $x\in V$. 
    Then, by the above, $h$ is a $3$-color assignment for
    $G$.

    \paragraph{(only-if)}  
        Suppose there is a $3$-color assignment $h$ for $G$. Let $\theta$ be the 
mapping defined as $\theta(x)=h(x)$ for $x \in V$ and $h(c)=c$
for $c \in C$.  As $h$ is a $3$-color assignment, $\body{\theta(\qr')} \subseteq
\body{\qr}$.  As there are no head variables, it readily follows that $\theta$ is a
simplification. Taking $\rho$ as the identity for $C$, it follows that
    $\body{\theta(\qr')} \subseteq \body{\qr} = \body{\rho(\qr)}$.

    \paragraph{Acyclicity of $\qr$.}

    By construction of $\qr$, $\vars{\body{\qr}} \subseteq C$. So, every
    hyperedge in $\hg_\qr$ is contained in the hyperedge that
    represents atom $\Fix(r, g, b)$, implying that all the other hyperedges can be removed
    immediately. Hereafter,  only one hyperedge remains (i.e, the hyperedge that represents
    $\Fix(r,g,b)$), implying that all the remaining nodes can be removed, which results in
    the empty hyperedge. 
    Hence, $\qr$ is acyclic.
\end{proof}


\begin{proposition} For Boolean CQs \qr and $\qr'$, where $\qr'$ is
    acyclic, it is \NP-hard to decide whether 
    condition (C3) holds.
    \label{prop:hc-com-acyclic1} \end{proposition}

\begin{proof} We give a reduction from the \NP-complete graph 3-colorability problem.  Let
    $G=(V,E)$ be an undirected graph, where $m=|E|$.  For ease of presentation, we assume
    that every edge is represented by only one tuple in $E$, i.e., that $(x,y) \in E$
    implies $(y,x) \not \in E$ \footnote{Notice that this requires to choose a direction for
    every edge, but that no information gets lost. The choice itself is of no further concern
and can be implemented in an arbitrary fashion.}. 

    Further, we consider a labeling function
    $\ell$ that assigns to every edge in $G$ a unique label in $\uvar\setminus
    (V\cup C)$. Let $Z = \img{\ell}=\{z_1,\ldots,z_m\}$. Let $W = \{w_{z,i} \mid
    z \in Z, i\in \{1,\ldots,10m\}$. We assume $W$ is disjoint from $V$, $C$, and $Z$.
    In the construction we represent edges of $G$ as
    atoms over the ternary relation $E$, where the first variable denotes its associated
    label, the second and third variable denote its end nodes.

    Based on $G$, we construct queries $\qr'$ and $\qr$ as follows:
    \begin{align*}
        \qr' & :() \leftarrow \bigwedge_{(c,d) \in E_C, z \in Z} E(z,c,d),
        \bigwedge_{i\in\{1, \ldots, m-1\}} \Fix(z_{i}, z_{i+1}, r, g, b). \\
        \qr & :() \leftarrow \bigwedge_{(x,y) \in E}E({\ell(x,y)},x, y),
        \bigwedge_{z\in Z,i\in \{1, 3, 5, 7, 9\}}E(z,w_{z,i}, w_{z,i+1}),
        \bigwedge_{i\in\{1, \ldots, m-1\}} \Fix(z_{i}, z_{i+1}, r, g, b).
    \end{align*}

    Notice that $\qr'$ represents for every edge in $G$ all the potentially valid
    color-assignments within the range $C$, implying $6$ $E$-atoms for each edge in $G$,
    whereas $\qr$ has for every edge in $G$, one $E$-atom that represents the edge and $5$
    additional $E$-atoms, containing unique variables on the end-node positions.
    Hereafter, we will refer to the latter as \emph{free} $E$-atoms.

    \paragraph{Claim.} There is a simplification $\theta$ for $\qr'$ and a substitution
    $\rho$ for $\qr$ such that $\body{\theta(\qr')} \subseteq \body{\rho(\qr)}$ if and only
    if there is a 3-color assignment for $G$.

    \paragraph{(if)}
    Let $h$ be a valid $3$-color assignment for $G$. We show that there is a
    simplification $\theta$ for $\qr'$ and substitution $\rho$ for $\qr$ such that
    $\body{\theta(\qr')}
    \subseteq \body{\rho(\qr)}$. Let $\theta$ be the identity simplification for $\qr'$.   Let $\rho'$ be the partial substitution for $\qr$ that is the identity on edge labels and colors, and maps the variables in $V$ onto variables in $C$ consistent with $h$. That is, $\rho'(x) = x$ for $x \in C\cup Z$ and $\rho'(x)=h(x)$ for $x \in V$. Now, $\rho'$ already guarantees that every extension of $\rho'$ satisfies the  containment property $\body{\qr'} \subseteq \rho'(\body{\qr})$ for all the $\Fix$-atoms
    in $\qr'$ and for exactly $m$ $E$-atoms in $\qr'$ (all having a distinct edge label).  Next, we extend $\rho'$ to a complete substitution $\rho$ by mapping the variables in $W$
    in such a way that the $5m$ free $E$-atoms are mapped on exactly the $5m$ remaining
    $E$-atoms in $\qr'$. Notice that such an extension exists because the variables in $W$
    are all unique for $\qr$. 

    \paragraph{(only-if)}
    Let $\theta$ be a simplification for $\qr'$ and $\rho$ a substitution for $\qr$, where
    $\body{\theta(\qr')} \subseteq \body{\rho(\qr)}$. We show that there is a valid
    $3$-color assignment for $G$. The proof proceeds under the assumption that
    \begin{align*}
        \text{$\body{\theta(\qr')} = \body{\qr'}$.} \hspace{4ex}(\dagger)
    \end{align*}
    We prove below that $(\dagger)$ holds. Assuming $(\dagger)$, it follows that
    $\body{\qr'} \subseteq \body{\rho(\qr)}$.  Let $h$ be the mapping defined by $h(x) = \rho(x)$ for every $x \in V$.
    To show that $h$ is a valid $3$-color mapping for $G$, we argue that
    $\img{h} \subseteq C$ and $h(x) \ne h(y)$ for every two adjacent nodes $x,y$ in $G$.
    We observe that for $\rho$ to satisfy $\body{\qr'} \subseteq \body{\rho(\qr)}$, $\rho$ must be such that every $E$-atom in the body of $\qr'$ is mapped onto an $E$-atom in the body of $\qr$. Then, by construction of $\qr'$ and $\qr$ (both contain precisely six $E$-atoms per edge label $z$) it  follows that every variable in $V$ is
    mapped by $\rho$ onto a variable in $C$. In particular, for every edge-representing
    atom in $\qr$, $\rho(E(z, x, y)) = E(z', c, d)$, where $c \ne d$.  So,
    $\img{h} \subseteq C$, and $h(x) \ne h(y)$ for every two adjacent nodes in
    $V$, implying $h$ to  be a valid $3$-color mapping for $G$.

    It remains to show $(\dagger)$. We actually prove a stronger result, namely that $\theta$ must be the identity simplification.
    By definition of simplification, $\body{\theta(\qr')} \subseteq \body{\qr'}$.  In
    particular, $\theta(\Fix(z_{i}, z_{i+1}, r, g, b)) \subseteq \body{\qr'}$ for every $i
    \in \{1, \ldots, m-1\}$. As the $\Fix$-atoms in $\qr'$ all agree on the last three  variables, it follows that $\theta$ is the identity on the
    variables $r, g, b$. It remains to show that $\theta$ is also the identity on the edge labels.  Towards a contradiction suppose that one of the edge labels in $\qr'$ is mapped onto another
    variable by $\theta$. Then, by construction of $\qr'$, it must be that this other
    variable is also an edge-label. Therefore, let $z_i$ be such a variable, where
    $\theta(z_i) = z_j$, for $i,j \in \{1, \ldots, m\}$ and $i \ne j$.
    Hereafter, by $\Fix_k$ we denote the atom $\Fix(z_k, z_{k+1}, r, g, b)$, where $1\le k \le
    {m-1}$.

    We distinguish two cases:
    
    \begin{enumerate}
        
        \item \underline{Case ($i < j$):}
            By the assumption $\theta(z_i) = z_j$ and the definition of simplification
            (i.e., $\theta(\fix_i) \in \body{\qr'}$), it must be that $\theta(\Fix_i) = \Fix_j$
            (because $\fix_j$ is the only atom in $\qr'$ where variable $x_j$ occurs on the
            first position).  So, $\theta(x_{i+1}) = \theta(x_{j+1})$. 
            We can repeat the above argument to show that $\theta(\Fix_{i+1}) = \Fix_{j+1}$,
            and in general, $\theta(\Fix_{i+k}) = \Fix_{j+k}$ for every positive integer $k
            \le m-j-1$. Notice that $\theta(\Fix_{m-j-1}) = \Fix_{m-1}$, i.e., it is the
            last atom of the chain. Further, $\theta(z_{m-j}) = z_{m}$. 
            Now, because $\theta(\body{\qr'}) \subseteq \body{\qr'}$, it must be that
            $\theta(\Fix_{m-j-1}) \subseteq \body{\qr'}$, requiring that there is an atom
            in $\body{\qr'}$ with $z_{m}$ on the first position. By construction of $\qr'$
            there is no such atom. 

        \item \underline{Case ($i > j$):} The proof is analogous to the proof for case
            $(1)$, except that we show an argument in the reverse direction, towards an
            atom where $z_1$ is on the second position.  
    
    \end{enumerate}
    
    \paragraph{Acyclicity of $\qr$.}
    Again, acyclicity can be verified by a
    straightforward GYO reduction on the hypergraph $\hg_{\qr'}$.
    Notice that all the hyperedges that represent an $E$-atom for $\qr'$
    are contained in the hyperedges that represent the $\Fix$-atoms for a matching
    edge label. So, all the hyperedges representing an $E$-atom can be removed immediately. Now, the proof proceeds by induction on the number of $\Fix$-atoms.  For the base case, assume there is only one $Fix$-atom. Then 
    all variables occur in only one hyperedge and thus can be removed which results in the empty hyperedge.  Next, assume that $\qr'$ is acyclic when it has a chain of at most $i$ $\Fix$-atoms. For the induction step we assume that there are $i+1$ $\Fix$-atoms.
    W.l.o.g., we assume that the first variable in the chain is $z_1$, and the last is $z_{i+2}$.  Notice that $z_1$ is in only one hyperedge, namely the
    hyperedge that represents $Fix(z_1, z_2, r, g, b)$. After removing $z_1$, this hyperedge is entirely contained by the hyperedge representing the next $\Fix$-atom, i.e., $\Fix(z_2,
    z_2, r, g, b)$. Consequently, the hyperedge that represents $\Fix(z_1, z_2, r, g, b)$ is
    eliminated. The resulting hypergraph now represents a chain of $i$ $\Fix$-atoms.
    Hence, the result follows from the induction hypothesis.
\end{proof}

\begin{remark} As Proposition~\ref{prop:hc-com-acyclic2} shows NP-completeness when $\qr$
    is acyclic and Proposition~\ref{prop:hc-com-acyclic1} shows NP-completeness when
    $\qr'$ is acyclic, it is an interesting question whether it is still NP-complete to
    decide condition (C3) 
        if both $\qr$ and $\qr'$ are acyclic. 
    When allowing arbitrary arity relations, acyclicity is easily achieved by the use of one atom that contains all variables of the query. Particularly, one can add to
    $\qr$ in the reduction for Proposition~\ref{prop:hc-com-acyclic1} an
    atom $A$ that contains every variable of $\qr$.
    Then, the proof remains valid as-is, but both $\qr'$ and $\qr$ are acyclic.
    Nevertheless, under bounded-arity database schemas, the problem remains open.
\end{remark}

